\documentclass[journal]{IEEEtran}
\usepackage{amsfonts}
\usepackage{amssymb}
\usepackage{amsthm}
\usepackage{amsmath,amsfonts,amssymb}
\usepackage[dvips]{graphicx}
\usepackage{verbatim}
\usepackage{setspace}
\usepackage{bm}
\usepackage[ruled,vlined]{algorithm2e}
\usepackage{cite}

\usepackage{changepage}
\usepackage{pdfpages}
\usepackage{color}
\newtheorem{theorem}{Theorem}
\newtheorem{lemma}{Lemma}

\newcommand{\figwidth}{8.8}

\begin{document}
\title{Joint Power Allocation and Beamforming for Non-Orthogonal Multiple Access (NOMA) in 5G Millimeter-Wave Communications}


\author{Zhenyu Xiao,~\IEEEmembership{Senior Member,~IEEE,}
        Lipeng Zhu,
        Jinho Choi,~\IEEEmembership{Senior Member,~IEEE,}
 Pengfei Xia,~\IEEEmembership{Senior Member,~IEEE}
and Xiang-Gen Xia,~\IEEEmembership{Fellow,~IEEE}
\thanks{This work was supported in part by the National Natural Science Foundation of China (NSFC) under grant Nos. 61571025 and 91538204, and in part by the Foundation for Innovative Research Groups of the National Natural Science Foundation of China under grant No. 61521091.}
\thanks{L. Zhu and Z. Xiao are with the School of Electronic and Information Engineering, Beihang University, Beijing 100191, China; Beijing Key Laboratory for Network-Based Cooperative Air Traffic Management, and Beijing Laboratory for General Aviation Technology, Beijing 100191, China; Collaborative Innovation Center of Geospatial Technology, Wuhan 430079, China.}
\thanks{J. Choi is with the School of Electrical Engineering and Computer Science, Gwangju Institute of Science and Technology (GIST), Korea.}
\thanks{P. Xia is with the School of Electronics and Information Engineering and the Key Laboratory of Embedded System and Service Computing, Tongji University, Shanghai 200092, P.R. China.}
\thanks{X.-G. Xia is with the College of Information Engineering, Shenzhen University, Shenzhen, China, and the Department of Electrical and Computer Engineering, University of Delaware, Newark, DE 19716, USA.}
}

\maketitle




\begin{abstract}
In this paper we explore non-orthogonal multiple access (NOMA) in millimeter-wave (mmWave) communications (mmWave-NOMA). In particular, we consider a typical problem, i.e., maximization of the sum rate of a 2-user mmWave-NOMA system. In this problem, we need to find the beamforming vector to steer towards the two users simultaneously subject to an analog beamforming structure, while allocating appropriate power to them. As the problem is non-convex and may not be converted to a convex problem with simple manipulations, we propose a suboptimal solution to this problem. The basic idea is to decompose the original joint beamforming and power allocation problem into two sub-problems which are relatively easy to solve: one is a power and beam gain allocation problem, and the other is a beamforming problem under a constant-modulus constraint. Extension of the proposed solution from 2-user mmWave-NOMA to more-user mmWave-NOMA is also discussed. Extensive performance evaluations are conducted to verify the rational of the proposed solution, and the results also show that the proposed sub-optimal solution achieve close-to-bound sum-rate performance, which is significantly better than that of time-division multiple access (TDMA).
\end{abstract}




%
\IEEEpeerreviewmaketitle

\section{Introduction}
\IEEEPARstart{A}{s} the fifth generation (5G) wireless mobile communication comes closer in the past a few years, the requirements gradually become clearer, and among them the large aggregate capacity is one of the most critical issues \cite{andrews2014will}. In the approach to high aggregate capacity, millimeter-wave (mmWave) communication has been considered as one of the major candidate technologies \cite{andrews2014will,niu2015survey,rapp2013mmIEEEAccess}, in addition to ultra-densification of cells and massive multiple-input multiple-output (MIMO). Indeed, due to the abundant frequency spectrum resource, mmWave communication promises a much higher capacity than the legacy low-frequency (i.e., micro-wave band) mobile communications. In addition, subject to high propagation loss, mmWave communication is quite suitable for small cell scenarios, which is preferable in 5G.

On the other hand, when applying mmWave communication to mobile cellular, its benefit will highly rely on multiple access strategies. Subject to the limited number of resource blocks, the existing time/frequency/code division multiple access (TDMA/FDMA/CDMA) may face stringent challenges in supporting a great increase of users in the future 5G paradigm, which is supposed to connect massive users and devices \cite{andrews2014will}. Moreover, due to the diversity of users in 5G cellular, the requirements on data rate will be quite different between different users. To allocate resource with the unit of a single resource block to users with varying rate requirements may be a waste of resource. In brief, the inefficient orthogonal multiple access (OMA), i.e., TDMA/FDMA/CDMA, may offset the benefit of mmWave communication in the future 5G cellular. Different from these OMA schemes, the non-orthogonal multiple access (NOMA) strategy is able to support multiple users in the same (time/frequency/code) resource block realized by superposition coding in the power domain \cite{ding2014performance,saito2013non,Choi2014NOMA,Ding2015Cooperative,Dai2015NOMA5G,Benjebbour2013ConceptNOMA,Saito2013syslevl,Ding2017random,Daill2017}.
By exploiting corresponding successive interference cancellation (SIC) in the power domain at receivers, multiple users can be distinguished from each other, thus both the number of users and the spectrum efficiency can be increased.

{ The use of NOMA also harvests benefits in mmWave communications. When mmWave communication is used for 5G, a big challenge is to support massive users and devices. However, subject to the hardware cost, the number of radio-frequency (RF) chains in an mmWave device is usually much smaller than that of antennas \cite{Xia_2011_60GHz_Tech,xia_2008_prac_ante_traning,wang_2009_beam_codebook,alkhateeb2014mimo,roh2014millimeter,sun2014mimo}. As a result, the maximal number of users that can be served within one time/frequency/code resource block is very limited, i.e., no larger than the number of RF chains. In such a case, NOMA is with significance for mmWave communication to greatly increase the number of users, and meanwhile increase the usage efficiency of the acquired spectrum to support the exponential traffic growth. On the other hand, the highly direction feature of mmWave propagation makes the users' channels (along the same or similar direction) highly correlated, which facilitates the integration of NOMA in mmWave communication, i.e., mmWave-NOMA \cite{Ding2017random,Daill2017}.

An intrinsic difference between mmWave-NOMA and conventional NOMA is that, beamforming with a large antenna array is usually adopted in mmWave-NOMA \cite{Xia_2011_60GHz_Tech,xia_2008_prac_ante_traning,wang_2009_beam_codebook,alkhateeb2014mimo,roh2014millimeter,sun2014mimo}, which means that power allocation intertwines with beamforming. In \cite{Ding2017random}, a pioneer work of combining NOMA in mmWave communications, random steering single-beam forming was adopted, which can work only in a special case that the NOMA users are close to each other. In \cite{Daill2017}, the new concept of beamspace MIMO-NOMA with a lens-array hybrid beamforming structure was firstly proposed to use multi-beam forming to serve multiple NOMA users with arbitrary locations, thus the limit that the number of supported users cannot be larger than the number of RF chains can be broken. However, the power allocation problem is studied under fixed beam pattern when lens array is considered. In \cite{Fang2016rescealloc}, sub-channel assignment and power allocation are optimized to maximize the energy efficiency for a downlink NOMA network, while beamforming is not included.} Different from \cite{Ding2017random,Daill2017,Fang2016rescealloc}, we consider joint power allocation and beamforming to maximize the sum rate of a 2-user mmWave-NOMA system using an analog beamforming structure with a phased array.


%

As the considered problem is non-convex and may not be converted to a convex problem with simple manipulations, to solve the problem directly by using the existing optimization tools is infeasible. On the other hand, to directly search the optimal solution is computationally prohibitive because the number of variables is large in general. In this paper, we propose a suboptimal solution to this problem. The basic idea is to decompose the original joint beamforming and power allocation problem into two sub-problems: one is a power and beam gain allocation problem, and the other is a beamforming problem under the CM constraint. Although the original problem is difficult to solve, the two sub-problems are relatively easy to solve, and thus we are able to obtain a sub-optimal solution. Extension of the proposed solution from 2-user mmWave-NOMA to more-user mmWave-NOMA is also discussed. Extensive performance evaluations are conducted to verify the rational of the proposed solution, and the results show that the proposed sub-optimal solution achieve close-to-bound sum-rate performance, which is significantly better than that of time-division multiple access (TDMA).

The rest of this paper is organized as follows. The system model and problem formulation are introduced in Section II. The solution is proposed in detail in Section III. Performance evaluation and comparison are conducted in Section IV. Lastly, the conclusions are drawn in Section V.

Symbol Notation: $a$ and $\mathbf{a}$ denote a scalar variable and a vector, respectively. $(\cdot)^{\rm{*}}$, $(\cdot)^{\rm{T}}$ and $(\cdot)^{\rm{H}}$ denote conjugate, transpose and conjugate transpose, respectively. $|\cdot|$ and $\|\cdot\|$ denote the absolute value and two-norm, respectively. $\mathbb{E}(\cdot)$ denotes the expectation operation. {$\mathrm{Re}(\cdot)$ denotes the real component of a complex number.} $[\mathbf{a}]_i$ denotes the $i$-th entry of $\mathbf{a}$.

\section{System Model and Problem Formulation}
%
%
%
%

\subsection{System model}
Without loss of generality, we consider a downlink multiuser scenario in this paper as shown in Fig. \ref{fig:system}, where a base station (BS) equipped with an $N$-element antenna array serves two users with a single antenna\footnote{In the case that the users also use an antenna array, Rx beamforming can be done first. Then the reception processing at each user can be seen equivalent to a single-antenna receiver, and the proposed solution in this paper can be used.}. At the BS, each antenna branch has a phase shifter and a power amplifier (PA) to drive the antenna. Generally, all the PAs have the same scaling factor. Thus, the beamforming vector, i.e., the antenna weight vector (AWV), has constant-modulus (CM) elements. The BS transmits a signal $s_{i}$ to User $i$ ($i=1,~2$), where $\mathbb{E}(\left | s_{i} \right |^{2})=1$, with transmission power $p_{i}$. The total transmission power is restricted to $P$. With 2-user NOMA, $s_{1}$ and $s_{2}$ are superimposed as
\begin{equation}
s=\sqrt{p_{1}}s_{1}+\sqrt{p_{2}}s_{2}
\end{equation}

\begin{figure}[t]
\begin{center}
  \includegraphics[width=\figwidth cm]{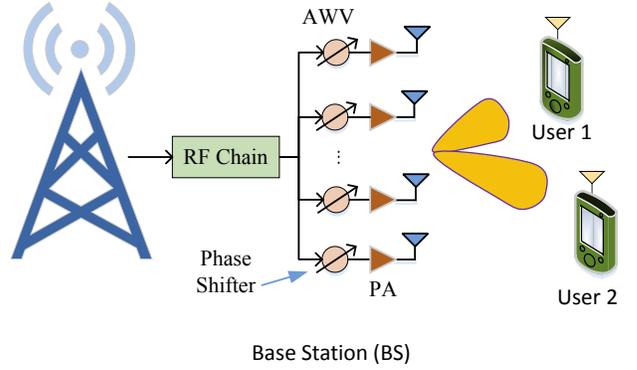}
  \caption{Illustration of an mmWave mobile cell, where one BS with $N$ antennas serves multiple users with one single antenna.}
  \label{fig:system}
\end{center}
\end{figure}

The received signals at User 1 and User 2 are
\begin{equation}
\left\{
   \begin{aligned}
y_{1}=\mathbf{h}_{1}^{\rm{H}}\mathbf{w}(\sqrt{p_{1}}s_{1}+\sqrt{p_{2}}s_{2})+n_{1} \\
y_{2}=\mathbf{h}_{2}^{\rm{H}}\mathbf{w}(\sqrt{p_{1}}s_{1}+\sqrt{p_{2}}s_{2})+n_{2}
   \end{aligned}
  \right.
  \end{equation}
where  ${\bf{h}}_{i}$ is channel response vectors between User $i$ and the BS, ${\bf{w}}$ denotes a CM beamforming vector with $|[{\bf{w}}]_k|=\frac{1}{\sqrt{N}}$ for $k=1,2,...,N$, and $n_{i}$ denotes the Gaussian white noise at User $i$ with power $\sigma^{2}$.

{
The channel between the BS and User $i$ is an mmWave channel. Subject to limited scattering in the mmWave band, multipath is mainly caused by reflection. As the number of the multipath components (MPCs) is small in general, the mmWave channel has directionality and appears spatial sparsity in the angle domain \cite{peng2015enhanced,wang2015multi,Lee2014exploiting,Gao2016ChannelEst,xiao2016codebook,alkhateeb2014channel}. Different MPCs have different angles of departure (AoDs).} Without loss of generality, we adopt the directional mmWave channel model assuming a uniform linear array (ULA) with a half-wavelength antenna space. Then an mmWave channel can be expressed as \cite{peng2015enhanced,wang2015multi,Lee2014exploiting,Gao2016ChannelEst,xiao2016codebook,alkhateeb2014channel}
\begin{equation} \label{eq_oriChannel}
\bar{\mathbf{h}}_{i}=\sum_{\ell=1}^{L_i}\lambda_{i,\ell}\mathbf{a}(N,\Omega_{i,\ell})
\end{equation}
where $\lambda_{i,\ell}$, $\Omega_{i,\ell}$ are the complex coefficient and cos(AoD) of the $\ell$-th MPC of the channel vector for User $i$, respectively, $L_i$ is the total number of MPCs for User $i$, ${\bf{a}}(\cdot)$ is a steering vector function defined as
\begin{equation} \label{eq_steeringVCT}
\mathbf{a}(N,\Omega)=[e^{j\pi0\Omega},e^{j\pi1\Omega},e^{j\pi2\Omega},\cdot\cdot\cdot,e^{j\pi(N-1)\Omega}]
\end{equation}
which depends on the array geometry. Let $\theta_{i,\ell}$ denote the real AoD of the $\ell$-th MPC for User $i$, then we have $\Omega_{i,\ell}=\cos(\theta_{i,\ell})$. Therefore, $\Omega_{i,\ell}$  is within the range $[-1, 1]$. For convenience and without loss of generality, in the rest of this paper, $\Omega_{i,\ell}$ is also called AoD.

For each user, the BS would perform beamforming toward the angle direction along the AoD of the strongest MPC to achieve a high array gain. In general, if there is no blockage between the BS and a user, the line-of-sight (LOS) component will be adopted for beamforming, as it has a much higher strength than the non line-of-sight (NLOS) components. If the LOS component is blocked, the strongest NLOS path would be selected for beamforming. Since the mmWave channel is spatially sparse, i.e., different MPCs have small mutual effects, we can obtain an effective channel model for the original channel model \eqref{eq_oriChannel} as
\begin{equation} \label{eq_effChannel}
\mathbf{h}_{i}=\lambda_{i}\mathbf{a}(N,\Omega_{i})
\end{equation}
where $\lambda_{i}=\lambda_{i,m_i}$ and $\Omega_{i}=\Omega_{i,m_i}$. Here $m_i$ denotes the index of the strongest MPC for User $i$. Since the effective channel model \eqref{eq_effChannel} is simpler, we adopt it in the derivation and analysis in this paper, while in the performance evaluations we also consider the original channel model in \eqref{eq_oriChannel}. Without loss of generality, we assume $|\lambda_{1}|\geq |\lambda_{2}|$, which means that the channel gain of User 1 is better.

{
\subsection{Decoding Order}
In the conventional NOMA with single-antenna BS and users, usually the information of the user with a lower channel gain is decoded first to maximize the sum rate. In contrast, in mmWave-NOMA, the decoding order depends on both channel gain and beamforming gain. Thus, there are two cases for the 2-user mmWave-NOMA system.

\emph{Case 1:} $s_{1}$ is decoded first. In such a case, User 1 directly decodes $s_1$ by treating the signal component of $s_2$ as noise. In comparison, User 2 first decodes $s_{1}$ and subtracts the signal component of $s_{1}$ from the received signal $y_{2}$; then it decodes $s_2$. Therefore, User 2 can decode $s_{2}$ without the interference from $s_{1}$.

With this decoding method, the achievable rates of User $i$ ($i=1,~2$), denoted by $R_{i}$, are represented as
\begin{equation}\label{eq_R1}
\left\{\begin{aligned}
R_{1}^{(1)}&=\log_{2}(1+ \frac{\left |\mathbf{h}_{1}^{\rm{H}}\mathbf{w} \right |^{2}p_{1}}{\left |\mathbf{h}_{1}^{\rm{H}}\mathbf{w} \right |^{2}p_{2}+\sigma^{2}}) \\
R_{2}^{(1)}&=\log_{2}(1+ \frac{\left |\mathbf{h}_{2}^{\rm{H}}\mathbf{w} \right |^{2}p_{2}}{\sigma^{2}})
\end{aligned}\right.
\end{equation}
Note that there is an implicit assumption for this result, i.e., the observed signal-to-interference plus noise ratio (SINR) at User 2 for decoding $s_1$ should be no less than that at User 1; otherwise $s_1$ cannot be correctly decoded at User 1 and the interference at User 2 (i.e., the signal component of $s_1$) cannot be completely removed. We call this implicit assumption an \emph{implicit SINR constraint}, and it is expressed as \label{eq_implicitcons1}
\begin{equation}
\frac{\left |\mathbf{h}_{2}^{\rm{H}}\mathbf{w} \right |^{2}p_{1}}{\left |\mathbf{h}_{2}^{\rm{H}}\mathbf{w} \right |^{2}p_{2}+\sigma^{2}} \geq
\frac{\left |\mathbf{h}_{1}^{\rm{H}}\mathbf{w} \right |^{2}p_{1}}{\left |\mathbf{h}_{1}^{\rm{H}}\mathbf{w} \right |^{2}p_{2}+\sigma^{2}}
\end{equation}
which is equivalent to $\left |\mathbf{h}_{2}^{\rm{H}}\mathbf{w} \right |^{2} \geq \left |\mathbf{h}_{1}^{\rm{H}}\mathbf{w} \right |^{2}$.

\emph{Case 2:} $s_2$ is decoded first. Similarly, with this decoding method, the achievable rates of User $i$ ($i=1,~2$), denoted by $R_{i}$, are represented as
\begin{equation}\label{eq_R2}
\left\{\begin{aligned}
R_{1}^{(2)}&=\log_{2}(1+ \frac{\left |\mathbf{h}_{1}^{\rm{H}}\mathbf{w} \right |^{2}p_{1}}{\sigma^{2}}) \\
R_{2}^{(2)}&=\log_{2}(1+ \frac{\left |\mathbf{h}_{2}^{\rm{H}}\mathbf{w} \right |^{2}p_{2}}{\left |\mathbf{h}_{2}^{\rm{H}}\mathbf{w} \right |^{2}p_{1}+\sigma^{2}})
\end{aligned}\right.
\end{equation}
The implicit SINR constraint is expressed as \label{eq_implicitcons2}
\begin{equation}
\frac{\left |\mathbf{h}_{1}^{\rm{H}}\mathbf{w} \right |^{2}p_{2}}{\left |\mathbf{h}_{1}^{\rm{H}}\mathbf{w} \right |^{2}p_{1}+\sigma^{2}}\geq \frac{\left |\mathbf{h}_{2}^{\rm{H}}\mathbf{w} \right |^{2}p_{2}}{\left |\mathbf{h}_{2}^{\rm{H}}\mathbf{w} \right |^{2}p_{1}+\sigma^{2}}
\end{equation}
which is equivalent to $\left |\mathbf{h}_{1}^{\rm{H}}\mathbf{w} \right |^{2} \geq \left |\mathbf{h}_{2}^{\rm{H}}\mathbf{w} \right |^{2}$.}

\subsection{Problem Formulation}
An immediate and basic problem is how to maximize the achievable sum rate of the two users provided that the channel is known \emph{a priori}. It is clear that if there are no minimal rate constraints for the two users, the achievable sum rate can be maximized by allocating all the power to User 1 and meanwhile beamforming toward User 1, whose channel gain is better. However, when there are minimal rate constraints for the two users, the power allocation intertwines with the beamforming design, which makes the problem complicated under the system setup.

{
We choose Case 2 as the decoding order; then the problem is formulated by
\begin{equation}\label{eq_problem}
\begin{aligned}
\mathop{\mathrm{Maximize}}\limits_{p_1,p_2,{\bf{w}}}~~~~ &R_{1}+R_{2}\\
\mathrm{Subject~ to}~~~ &R_{1} \geq r_{1}\\
&R_{2} \geq r_{2} \\
&\ p_{1}+p_{2}=P \\
&|[{\bf{w}}]_k|=\frac{1}{\sqrt{N}},~k=1,2,...,N\\
&\left |\mathbf{h}_{1}^{\rm{H}}\mathbf{w} \right |^{2} \geq \left |\mathbf{h}_{2}^{\rm{H}}\mathbf{w} \right |^{2}
\end{aligned}
\end{equation}
where $r_i$ denotes the minimal rate constraint for User $i$, $|[{\bf{w}}]_k|=\frac{1}{\sqrt{N}}$ is the CM constraint due to using the phase shifters in each antenna branch at the BS, $R_i=R_i^{(2)}$ as defined in \eqref{eq_R2}.

It is noteworthy that the decoding order can also be Case 1, and another problem can be formulated accordingly. The two problems with different decoding orders are similar to each other, which means that if an approach can be used to solve one of them, it can also be used to solve the other. For this reason, we adopt Case 2 as the decoding order, and propose solution to maximize the achievable sum rate. On the other hand, it is noted that since we assume that the channel gain of User 1 is better, usually Case 2 can achieve a better sum rate than Case 1. In fact, it can be proven that the optimal sum rate of Case 2 is better than that of Case 1 if $r_1=r_2$, which will be shown later.
}


\section{Solution of the Problem}
Clearly, directly solving Problem \eqref{eq_problem} by using the existing optimization tools is infeasible, because the problem is non-convex and may not be converted to a convex problem with simple manipulations. On the other hand, to directly search the optimal solution is also computationally prohibitive because the dimension is $(N+2)$, where $N$ is large in general. In this section, we propose a suboptimal solution to this problem. The basic idea is to decompose the original problem \eqref{eq_problem} into two sub-problems which are relatively easy to solve, and then we solve them one by one.

\subsection{Problem Decomposition}
Since power allocation intertwines with beamforming under the CM constraint, we first try to decompose them.
Let $c_{1}=\left|{{\bf{h}}}_{1}^{\rm{H}}\mathbf{w}\right|^{2}$ and $c_{2}=\left|{{\bf{h}}}_{2}^{\rm{H}}\mathbf{w}\right|^{2}$ denote the beam gains for User 1 and User 2, respectively. We have the following lemma.
\begin{lemma} With the ideal beamforming, the beam gains satisfy
\begin{equation}
\frac{c_{1}}{\left|\lambda_{1}\right|^{2}}+\frac{c_{2}}{\left|\lambda_{2}\right|^{2}}=N
\end{equation}
where $N$ is the number of antennas.
\end{lemma}

\begin{proof}
With the ideal beamforming, there is no side lobe, i.e., the beam gains along the directions of the two users are significant, while the beam gains along the other directions are zeros. Moreover, the beam pattern for each user is flat with a beam width of $2/N$, which is the same as that of an arbitrary steering vector shown in \eqref{eq_steeringVCT} \cite{TseFundaWC,xiao2016codebook}. Under such an ideal condition, the average power of an ideal beamforming vector $\bf{w}$ in the beam domain is
\begin{equation}
\begin{aligned}
&\frac{1}{2}\mathlarger{\int}_{-1}^{1}\left|\mathbf{a}(N,\Omega)^{\rm{H}}\mathbf{w}\right|^{2}d\Omega\\
=&\frac{1}{N}(\left|\mathbf{a}(N,\Omega_{1})^{\rm{H}}\mathbf{w}\right|^{2}+\left|\mathbf{a}(N,\Omega_{2})^{\rm{H}}\mathbf{w}\right|^{2})\\
=&\frac{1}{N}(\frac{c_{1}}{\left|\lambda_{1}\right|^{2}}+\frac{c_{2}}{\left|\lambda_{2}\right|^{2}})
\end{aligned}
\end{equation}

On the other hand, for an arbitrary AWV ${\bf{w}}$, we have
\begin{equation}
\begin{aligned}
&\frac{1}{2}\int_{ - 1}^1 {|{\bf{a}}{{(N,\Omega )}^{\rm{H}}}{\bf{w}}{|^2}d\Omega } \\
 =& \frac{1}{2}\int_{ - 1}^1 {\sum\limits_{m = 1}^N {{{[{\bf{w}}]}_m}{e^{ - j\pi (m - 1)\Omega }}\sum\limits_{n = 1}^N {[{\bf{w}}]_n^*{e^{j\pi (n - 1)\Omega }}} } d\Omega } \\
 =& \frac{1}{2}\int_{ - 1}^1 {\sum\limits_{m = 1}^N {\sum\limits_{n = 1}^N {{{[{\bf{w}}]}_m}[{\bf{w}}]_n^*{e^{ - j\pi (m - 1)\Omega }}{e^{j\pi (n - 1)\Omega }}} } d\Omega } \\
 =&\sum\limits_{m = 1}^N {{{[{\bf{w}}]}_m}[{\bf{w}}]_m^*}  + \\
 &~~~~\frac{1}{2}\sum\limits_{m = 1}^N {\sum\limits_{n = 1,~n \ne m}^N {{{[{\bf{w}}]}_m}[{\bf{w}}]_n^*\int_{ - 1}^1 {{e^{j\pi (n - m)\Omega }}d\Omega } } }\\
 =&\|{\bf{w}}\|^2
\end{aligned}
\end{equation}

Under the CM constraint, we have $\|{\bf{w}}\|^2=1$. Thus
\begin{equation}
\begin{aligned}
&\frac{1}{2}\mathlarger{\int}_{-1}^{1}\left|\mathbf{a}(N,\Omega)^{\rm{H}}\mathbf{w}\right|^{2}d\Omega=\frac{1}{N}(\frac{c_{1}}{\left|\lambda_{1}\right|^{2}}+\frac{c_{2}}{\left|\lambda_{2}\right|^{2}})=1\\
\Rightarrow&\frac{c_{1}}{\left|\lambda_{1}\right|^{2}}+\frac{c_{2}}{\left|\lambda_{2}\right|^{2}}=N
\end{aligned}
\end{equation}

\end{proof}

Based on Lemma 1, Problem \eqref{eq_problem} can be re-described as
\begin{equation} \label{eq_problem_sub1}
\begin{aligned}
\mathop{\mathrm{Maximize}}\limits_{p_1,p_2,c_1,c_2}~~~~&R_1+R_2\\
\mathrm{Subject~to} ~~~ &R_1 \geq r_{1}\\
&R_2 \geq r_{2} \\
&p_{1}+p_{2}=P\\
&\frac{c_{1}}{\left|\lambda_{1}\right|^{2}}+\frac{c_{2}}{\left|\lambda_{2}\right|^{2}}=N\\
&c_1 \geq c_2
\end{aligned}
\end{equation}
where $|{\bf{h}}_i^{\rm{H}}{\bf{w}}|^2$ is replaced by the beam gain $c_i$ ($i=1,2$). The CM constraint is not involved in the above problem \eqref{eq_problem_sub1}, but will be considered in the following beamforming sub-problem.


{
\begin{lemma} If the users share an identical minimal rate constraint, i.e. $r_1=r_2$, the optimal decoding order at the user side is the order of increasing $|\lambda_i|$. In particular, decoding User 2 first is the optimal order in our model.
\end{lemma}

\begin{proof}
With Case 1, the achievable rates of User $i~(i=1,2)$ are represented as
\begin{equation}\label{decoding_1}
\left\{\begin{aligned}
R_{1}^{(1)}&=\log_{2}(1+ \frac{c_{1}p_{1}}{c_{1}p_{2}+\sigma^{2}}) \\
R_{2}^{(1)}&=\log_{2}(1+ \frac{c_{2}p_{2}}{\sigma^{2}})
\end{aligned}\right.
\end{equation}
where $c_1 \leq c_2$. The other constraints are the same as those of Case 2, as shown in \eqref{eq_problem_sub1}.

With Case 2, the achievable rates of User $i~(i=1,2)$ are represented as
\begin{equation}\label{decoding_2}
\left\{\begin{aligned}
R_{1}^{(2)}&=\log_{2}(1+ \frac{c_{1}p_{1}}{\sigma^{2}}) \\
R_{2}^{(2)}&=\log_{2}(1+ \frac{c_{2}p_{2}}{c_{2}p_{1}+\sigma^{2}})
\end{aligned}\right.
\end{equation}
where $c_1\geq c_2$.

Assume that the optimal solution for Case 1 is $\{c_{1}^{\star}, c_{2}^{\star}, p_{1}^{\star}, p_{2}^{\star}\}$. We have $c_1^ \star  \le c_2^ \star$, and $\frac{{c_1^ \star }}{{|{\lambda _1}{|^2}}} + \frac{{c_2^ \star }}{{|{\lambda _2}{|^2}}} = N$ (Lemma 1). Hence, we can obtain

\begin{equation}
\begin{aligned}
&\left( {N - \frac{{c_2^ * }}{{|{\lambda _2}{|^2}}}} \right)|{\lambda _1}{|^2} \le c_2^ * \\
 \Rightarrow &N \le \frac{{c_2^ * }}{{|{\lambda _2}{|^2}}} + \frac{{c_2^ * }}{{|{\lambda _1}{|^2}}} \Rightarrow N \le \frac{{|{\lambda _1}{|^2} + |{\lambda _2}{|^2}}}{{|{\lambda _1}{|^2}|{\lambda _2}{|^2}}}c_2^ * \\
 \Rightarrow &N\left( {|{\lambda _1}{|^2} - |{\lambda _2}{|^2}} \right) \le \frac{{|{\lambda _1}{|^4} - |{\lambda _2}{|^4}}}{{|{\lambda _1}{|^2}|{\lambda _2}{|^2}}}c_2^ * \\
 \Rightarrow &\left( {N - \frac{{c_2^ * }}{{|{\lambda _2}{|^2}}}} \right)|{\lambda _1}{|^2} \le \left( {N - \frac{{c_2^ * }}{{|{\lambda _1}{|^2}}}} \right)|{\lambda _2}{|^2}
 \end{aligned}
\end{equation}

For Case 2, we choose $\{c_{1}=c_{2}^{\star},~p_{1}=p_{2}^{\star},~p_{2}=p_{1}^{\star}\}$. Thus we have

\begin{equation}
\begin{aligned}
c_2&=(N-\frac{c_{1}}{\left|\lambda_{1}\right|^{2}})\left|\lambda_{2}\right|^{2}=(N-\frac{c_{2}^{\star}}{\left|\lambda_{1}\right|^{2}})\left|\lambda_{2}\right|^{2}\\
   &\geq (N-\frac{c_{2}^{\star}}{\left|\lambda_{2}\right|^{2}})\left|\lambda_{1}\right|^{2}=c_{1}^{\star}
\end{aligned}
\end{equation}

Further on, we have
\begin{equation}
\left\{\begin{aligned}
R_{1}^{(2)}&=\log_{2}(1+ \frac{c_{2}^{\star}p_{2}^{\star}}{\sigma^{2}})=R_{2}^{(1)}\\
R_{2}^{(2)}&\geq \log_{2}(1+ \frac{c_{1}^{\star}p_{1}^{\star}}{c_{1}^{\star}p_{2}^{\star}+\sigma^{2}})=R_{1}^{(1)}
\end{aligned}\right.
\end{equation}
So $R_{1}^{(2)}+R_{2}^{(2)}\geq R_{1}^{(1)}+R_{2}^{(1)}$, i.e., decoding User 2 first is optimal if the users share an identical minimal rate constraint.
\end{proof}}

Next, we formulate the beamforming problem, i.e., to design $\bf{w}$ such that $|{\bf{h}}_i^{\rm{H}}{\bf{w}}|^2=c_i$ ($i=1,2$) under the CM constraint. However, under strict constraints $|{\bf{h}}_i^{\rm{H}}{\bf{w}}|^2=c_i$ and the CM constraint, an appropriate weighting vector $\bf{w}$ may not be found in most cases, and the problem is still difficult to solve due to the equality constraints. Therefore, we relax the equality constraints $|{\bf{h}}_i^{\rm{H}}{\bf{w}}|^2=c_i$ with inequality constraints $|{\bf{h}}_i^{\rm{H}}{\bf{w}}|^2\geq c_i$, and formulate the beamforming problem as follows:
\begin{equation}\label{eq_problem_sub2}
\begin{aligned}
\mathop{\mathrm{Minimize}}\limits_{{\bf{w}}}~~~~&\alpha\\
\mathrm{Subject~to} ~~~ &[\mathbf{w}]_{i}[\mathbf{w}]_{i}^{*}\leq\alpha; \ i=1,2,\cdots,N \\
           ~~~   &\left|\mathbf{h}_{1}^{\rm{H}}\mathbf{w}\right|\geq\sqrt {c_{1}} \\
           & \left|\mathbf{h}_{2}^{\rm{H}}\mathbf{w}\right|\geq\sqrt {c_{2}}
\end{aligned}
\end{equation}
where the CM constraint $|{\bf{w}}|=\frac{{\bf{1}}}{\sqrt{N}}$ is replaced by $[\mathbf{w}]_{i}[\mathbf{w}]_{i}^{*}\leq\alpha$, $\ i=1,2,\cdots,N$, and we want to minimize $\alpha$, such that the 2-norm of $\bf{w}$ is minimal.

Note that although Problem \eqref{eq_problem_sub2} is relatively easy to solve, the CM constraint is not necessarily satisfied by solving \eqref{eq_problem_sub2}. In fact, the target behind Problem \eqref{eq_problem_sub2} is to \emph{try to let the power of each antenna be the same while satisfying the gain constraints.} Hence, after solving Problem \eqref{eq_problem_sub2}, we need to normalize ${\bf{w}}$ to satisfy the CM constraint with the phases of its elements unchanged.

{
Finally, considering that the normalized ${\bf{w}}$ may not satisfy the strict gain constraint $|{\bf{h}}_i^{\rm{H}}{\bf{w}}|^2=c_i$ ($i=1,2$), we need to substitute it into the original problem, i.e., \eqref{eq_problem}, to reset the user powers. Then we obtain the final solution.
}

With the above manipulations, Problem \eqref{eq_problem} is decomposed into Problems \eqref{eq_problem_sub1} and \eqref{eq_problem_sub2}, which are independent power and beam gain allocation and beamforming sub-problems. Problem \eqref{eq_problem_sub1} is a relaxation of the original problem, since the CM beamforming is bypassed, which enlarges the feasible region. Problem \eqref{eq_problem_sub2} addresses the beamforming issue, but it is also a relaxation of the CM beamforming problem. Although the original problem is hard to solve, the two sub-problems are relatively easy to solve. Next, we will first solve Problem \eqref{eq_problem_sub1}, and obtain the optimal solution $\{c_{1}^{\star},~c_{2}^{\star},~p_{1}^{\star},~p_{2}^{\star}\}$. The optimal values of $c_{1}$ and $c_{2}$ are used as gain constraints in Problem \eqref{eq_problem_sub2}. Then we solve Problem \eqref{eq_problem_sub2} and obtain an appropriate ${\bf{w}}$. Afterwards, we normalize ${\bf{w}}$ to satisfy the CM constraint. Finally, substituting ${\bf{w}}$ into problem \eqref{eq_problem}, we obtain the final solution of power allocation $\{p_{1}, p_{2}\}$. Although the obtained $\{p_{1}, p_{2}, {\bf{w}}\}$ are not globally optimal, the achieved sum rate performance is close to the performance bound, as it will be shown later in Section V.

\subsection{Solution of the Power and Beam Gain Allocation Sub-Problem}

Note that based on Lemma 2, the optimal decoding is confirmed, so the expression of achievable rates are exclusive, i.e., $R_1=R_{1}^{(2)}, R_2=R_{2}^{(2)}$. The implicit SINR constraint $\frac{c_{1}p_{2}}{c_{1}p_{1}+\sigma^{2}}\geq\frac{c_{2}p_{2}}{c_{2}p_{1}+\sigma^{2}}$ is equivalent to $c_1\geq c_2$. In addition, there are two equality constraints in Problem \eqref{eq_problem_sub1}, namely $p_{1}+p_{2}=P$ and
$\frac{c_{1}}{\left|\lambda_{1}\right|^{2}}+\frac{c_{2}}{\left|\lambda_{2}\right|^{2}}=N$. Thus, we have $p_2=P-p_1$ and $c_2=(N-\frac{c_{1}}{\left|\lambda_{1}\right|^{2}})\left|\lambda_{2}\right|^{2}$. Substituting them into Problem \eqref{eq_problem_sub1}, we obtain
\begin{equation}\label{eq_problem_sub1_1}
\begin{aligned}
\mathop{\mathrm{Maximize}}\limits_{p_1,c_1}~~~ &f(c_{1},p_{1})=R_{1}+R_{2}\\
\mathrm{Subject~to} ~~~ &R_{1} \geq r_{1}\\
&R_{2} \geq r_{2} \\
&c_1\geq c_2
\end{aligned}
\end{equation}
where $R_1$ and $R_2$ become
\begin{equation}
\left\{\begin{aligned}
R_{1}&=\log_{2}(1+ \frac{c_{1}p_{1}}{\sigma^{2}}) \\
R_{2}&=\log_{2}\left(1+ \frac{(\left|\lambda_{2}\right|^{2}N-\frac{\left|\lambda_{2}\right|^{2}}{\left|\lambda_{1}\right|^{2}}c_{1})(P-p_{1})}{(\left|\lambda_{2}\right|^{2}N-\frac{\left|\lambda_{2}\right|^{2}}{\left|\lambda_{1}\right|^{2}}c_{1})p_{1}+\sigma^{2}}\right)
\end{aligned}\right.
\end{equation}

The objective function $f(c_{1},p_{1})$ with two variables $c_{1}$ and $p_{1}$ is given by

\begin{equation}
\begin{aligned}
f(c_{1},p_{1})=&\log_{2}(1+ \frac{c_{1}p_{1}}{\sigma^{2}})+\\
&\log_{2}\left(1+ \frac{(\left|\lambda_{2}\right|^{2}N-\frac{\left|\lambda_{2}\right|^{2}}{\left|\lambda_{1}\right|^{2}}c_{1})(P-p_{1})}{(\left|\lambda_{2}\right|^{2}N-\frac{\left|\lambda_{2}\right|^{2}}{\left|\lambda_{1}\right|^{2}}c_{1})p_{1}+\sigma^{2}}\right)
\end{aligned}
\end{equation}

{It is known that the maximum point of a continuous function in a bounded closed set is either an extreme point or located in the boundary.} Thus, to solve Problem \eqref{eq_problem_sub1_1}, we first obtain the extreme points of the objective function. If the three inequality constraints are satisfied at any one of these extreme points, it may be an optimal solution. Otherwise the optimal solution should be within the boundary defined by the three inequality constraints. Hence, we start from obtaining the stationary points of the objective function. Note that a stationary point may not be necessarily an extreme point; it may be a saddle point.

Let the gradient of $f(c_{1},p_{1})$ be zero, i.e.,
\begin{equation}
  (\frac{\partial f}{\partial c_{1}},\frac{\partial f}{\partial p_{1}})=(0,0)
\end{equation}
Then, we obtain one stationary point $(c_{1m},p_{1m})$:
\begin{equation} \label{eq_stationaryPoint}
\left\{\begin{aligned}
c_{1m}&=\frac{\left|\lambda_{1}\right|^{2}\left|\lambda_{2}\right|^{2}N}{\left|\lambda_{1}\right|^{2}+\left|\lambda_{2}\right|^{2}}\\
  p_{1m}&= \frac{\left|\lambda_{2}\right|^{2}(\left|\lambda_{1}\right|^{2}+\left|\lambda_{2}\right|^{2})P\sigma^{2}}{\left|\lambda_{1}\right|^{4}\left|\lambda_{2}\right|^{2}
  NP+(\left|\lambda_{1}\right|^{2}+\left|\lambda_{2}\right|^{2})^{2}\sigma^{2}}
\end{aligned}\right.
\end{equation}

After some derivations the objective function at this stationary point writes
\begin{equation} \label{eq_stationaryValue}
  f(c_{1m},p_{1m})=\log_{2}\left(1+\frac{\left|\lambda_{1}\right|^{2}\left|\lambda_{2}\right|^{2}NP}{(\left|\lambda_{1}\right|^{2}+\left|\lambda_{2}\right|^{2})\sigma^{2}}\right)
\end{equation}

{At this moment, we do not know whether this stationary point is an extreme point or just a saddle point. Thus, we can examine the values of the functions at some other points. If there are values at the chosen points larger than $f(c_{1m},p_{1m})$ and there are also values smaller than $f(c_{1m},p_{1m})$, then this stationary point must be a saddle point. For simplicity, we may choose the intersection points of the boundaries, i.e, the maximum value and minimum value of $c_1, p_1$, respectively, namely $(0,0)$, $(N|\lambda_1|^2,P)$, $(N|\lambda_1|^2,0)$ and $(0,P)$.}

The values of the objective function at point $(0,0)$ and $(N,P)$ are
\begin{equation}
  f(0,0)=\log_{2}(1+\frac{\left|\lambda_{2}\right|^{2}NP}{\sigma^{2}})
\end{equation}
\begin{equation}
  f(N|\lambda_1|^2,P)=\log_{2}(1+\frac{\left|\lambda_{1}\right|^{2}NP}{\sigma^{2}})
\end{equation}
and the values of the objective function at point $(N,0)$ and $(0,P)$ are
\begin{equation}
  f(N|\lambda_1|^2,0)=0
\end{equation}
\begin{equation}
  f(0,P)=0
\end{equation}

Clearly, $f(N,0)$ and $f(0,P)$ are smaller than $f(c_{1m},p_{1m})$, while $f(0,0)$ and $f(N,P)$ are greater than $f(c_{1m},p_{1m})$. Hence, the point $(c_{1m},p_{1m})$ is just a saddle point instead of an optimum of the objective function. As a result, the objective function does not have an extreme point. Hence, the optimal solution of Problem \eqref{eq_problem_sub1_1} locates within the boundary of the feasible region defined by one of the inequality constraints. As there are three inequality constraints in Problem \eqref{eq_problem_sub1_1}, the feasible region is enclosed by three boundaries corresponding to the three inequality constraints, respectively, as illustrated in Fig. \ref{feasible_region}. The optimal solution may locate within any one of them. Thus, we discuss the possible cases here.

\begin{figure}[t]
\begin{center}
  \includegraphics[width=7 cm]{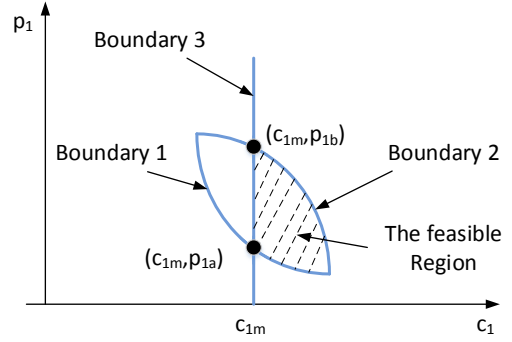}
  \caption{Illustration of the feasible region of Problem \eqref{eq_problem_sub1_1}.}
  \label{feasible_region}
\end{center}
\end{figure}

\emph{Case 1:} If the optimal solution is within the boundary $c_1=c_2$ (Boundary 3 in Fig. \ref{feasible_region}), we have $c_1=c_{1m}$ as shown in \eqref{eq_stationaryPoint}, and in such a case the objective function $f(c_1,p_1)$ does not depend on $p_1$, i.e., $f(c_1,p_1)=f(c_{1m},p_{1m})$ as shown in \eqref{eq_stationaryValue}, provided that $(c_1,p_1)$ satisfies the two rate constraints.

\emph{Case 2:} If the optimal solution is within the boundary $R_{1}=r_{1}$ (Boundary 1 in Fig. \ref{feasible_region}), we have $p_{1}= \frac{(2^{r_{1}}-1)\sigma^{2}}{c_{1}}$. Substituting it into the objective function $f(c_1,p_1)$, we find that the objective function now has only one variable $c_1$. By letting the derivative of the objective function with respect to $c_1$ equal zero, and solving the equation, we can obtain two roots, a positive one and a negative one. Clearly the negative one does not satisfy the condition that the beam gain $c_1$ is positive. Thus, the obtained positive root is the optimal value of $c_1$. Then we achieve the following optimal solution:
\begin{equation}
\left\{\begin{aligned}
c_{1,1}&=\frac{\left|\lambda_{1}\right|^{2}\big{[}(2^{1+r_{1}}-2)\left|\lambda_{2}\right|^{2}NP-2\sqrt{G}\big{]}}{2P\big{[}
(2^{r_{1}}-1)\left|\lambda_{2}\right|^{2}-\left|\lambda_{1}\right|^{2}\big{]}}\\
  p_{1,1}&= \frac{(2^{r_{1}}-1)\sigma^{2}}{c_{1,1}}
\end{aligned}\right.
\end{equation}
where
\begin{equation}
\begin{aligned}
G&=(2^{r_{1}}-1)\left|\lambda_{1}\right|^{2}\left|\lambda_{2}\right|^{2}N^{2}P^{2}+\\
&\big{[}(2^{r_{1}}-1)\left|\lambda_{1}\right|^{2}+(2^{1+r_{1}}-2^{2r_{1}}-1)\left|\lambda_{2}\right|^{2}
\big{]}NP\sigma^{2}.
\end{aligned}
\end{equation}

It is noteworthy that Case 2 implicitly requires $c_{1,1}\geq c_{1m}$, such that $f(c_{1,1},p_{1,1})\geq f(c_{1m},p_{1a})=f(c_{1m},p_{1m})$, where $(c_{1m},p_{1a})$ is the intersection point of Boundary 3 and Boundary 1 as shown in Fig. \ref{feasible_region}. Otherwise if $c_{1,1}< c_{1m}$, in the region of $c_1\geq c_{1m}$ within Boundary 1, $f(c_1,p_1)$ is monotonically descending as $c_1$. In such a case, $f(c_{1m},p_{1m})=f(c_{1m},p_{1a})\geq f(c_{1},p_{1})$, which means that the optimal solution locates within Boundary 3 instead of Boundary 1.

\emph{Case 3:} If the optimal solution is within the boundary $R_{2}=r_{2}$ (Boundary 2 in Fig. \ref{feasible_region}), analogously, we can obtain the following optimal solution:
\begin{equation}
\left\{\begin{aligned}
  c_{1,2}&=\left|\lambda_{1}\right|^{2}\left(N-\frac{\sqrt{(2^{r_{2}}-1)NP\sigma^{2}}}{\left|\lambda_{2}\right|P}\right)\\
  p_{1,2}&= \frac{(N\left|\lambda_{1}\right|^{2}-c_{1})\left|\lambda_{2}\right|^{2}P-(2^{r_{2}}-1)\left|\lambda_{1}\right|^{2}\sigma^{2}}{2^{r_{2}}\left|\lambda_{2}\right|^{2}
  (N\left|\lambda_{1}\right|^{2}-c_{1})}
\end{aligned}\right.
\end{equation}

Similarly, Case 3 implicitly requires $c_{1,2}\geq c_{1m}$; otherwise the optimal solution locates within Boundary 3 instead of Boundary 2.

%

In summary, there are four cases in total to determine the optimal solution $(c_1^\star,p_1^\star)$:

(i) If $c_{1,1}<c_{1m}$ and $c_{1,2}<c_{1m}$, $f^\star(c_1,p_1)=f(c_{1m},p_{0})=f(c_{1m},p_{1m})$, and $(c_1^\star,p_1^\star)=({c_{1m},p_{0}})$, where $p_{0}$ can be any value satisfying the two rate constraints.

(ii) If $c_{1,1}\geq c_{1m}$ and $c_{1,2}<c_{1m}$, $f^\star(c_1,p_1)=f(c_{1,1},p_{1,1})$, and $(c_1^\star,p_1^\star)=({c_{1,1},p_{1,1}})$.

(iii) If $c_{1,1}< c_{1m}$ and $c_{1,2}\geq c_{1m}$, $f^\star(c_1,p_1)=f(c_{1,2},p_{1,2})$, and $(c_1^\star,p_1^\star)=({c_{1,2},p_{1,2}})$.

(iv) If $c_{1,1}>c_{1m}$ and $c_{1,2}>c_{1m}$,
\begin{equation}
f^\star(c_1,p_1)=\max([f(c_{1,1},p_{1,1}),f(c_{1,2},p_{1,2})])
\end{equation}
The corresponding optimal solution is
$(c_1^\star,p_1^\star)=({c_{1,1},p_{1,1}})$ if $f(c_{1,1},p_{1,1})\geq f(c_{1,2},p_{1,2})$, or $(c_1^\star,p_1^\star)=({c_{1,2},p_{1,2}})$ if $f(c_{1,1},p_{1,1})< f(c_{1,2},p_{1,2})$.

The other two parameters are $p_2^\star=P-p_1^\star$ and $c_2^\star=(N-\frac{c_{1}^\star}{\left|\lambda_{1}\right|^{2}})\left|\lambda_{2}\right|^{2}$.

{Hereto, we have solved the power and beam gain allocation sub-problem, i.e., we have found the optimal solution of Problem \eqref{eq_problem_sub1} and obtained $\{c_{1}^\star,c_{2}^\star,p_1^\star,p_2^\star\}$ under the assumption of ideal beamforming, i.e., we assume Lemma 1 holds. However, $\{c_{1}^\star,c_{2}^\star,p_1^\star,p_2^\star\}$ may not be an optimal solution of the original problem, i.e., Problem \eqref{eq_problem}, because a beamforming vector with beam gains $\{c_{1}^\star,c_{2}^\star\}$ may not be found under the CM constraint. Hence, we say the optimal achievable sum rate of Problem \eqref{eq_problem_sub1} is an \emph{upper bound} of that of the original problem.}

\subsection{Solution of the Beamforming Sub-Problem}
In this section, we solve the beamforming sub-problem, i.e., we solve Problem \eqref{eq_problem_sub2} to design an appropriate ${\bf{w}}$ to realize the user beam gains $c_{1}^\star$ and $c_{2}^\star$.

Define $b_{1}=\frac{c_{1}^\star}{\left|\lambda_{1}\right|^{2}},b_{2}=\frac{c_{2}^\star}{\left|\lambda_{2}\right|^{2}}$. Problem \eqref{eq_problem_sub2} can be rewritten to
\begin{equation} \label{eq_problem_sub2_1}
\begin{aligned}
\mathop{\mathrm{Minimize}}\limits_{{\bf{w}}}~~~~&\alpha\\
\mathrm{Subject~to} ~~~ &[\mathbf{w}]_{i}[\mathbf{w}]_{i}^{*}\leq\alpha; \ i=1,2,\cdots,N \\
           ~~~   &\left|\mathbf{a}_{1}^{\rm{H}}\mathbf{w}\right|\geq\sqrt {b_{1}} \\
           & \left|\mathbf{a}_{2}^{\rm{H}}\mathbf{w}\right|\geq\sqrt {b_{2}}
\end{aligned}
\end{equation}
where ${\mathbf{a}}_{i}\triangleq \mathbf{a}(N,\Omega_i)$ for $i=1,2$.

It is clear that an arbitrary phase rotation can be added to the vector $\mathbf{w}$ in Problem \eqref{eq_problem_sub2_1} without affecting the beam gains. Thus, if $\mathbf{w}$ is optimal, so is $\mathbf{w}e^{j\phi}$, where $\phi$ is an arbitrary phase within $[0,2\pi)$. Without loss of generality, we may then choose $\phi$ so that ${\bf{h}}_{1}^{\rm{H}}\mathbf{w}$ is real. Problem \eqref{eq_problem_sub2_1} is tantamount to
\begin{gather} \label{eq_problem_sub2_1a}
\begin{aligned}
\mathop{\mathrm{Minimize}}\limits_{{\bf{w}}}~~~~&\alpha\\
\mathrm{Subject~to} ~~~ &[\mathbf{w}]_{i}[\mathbf{w}]_{i}^{*}\leq\alpha; \ i=1,2,\cdots,N \\
            &~{\mathrm{Re}(\mathbf{a}_{1}^{\rm{H}}\mathbf{w}) \geq\sqrt {b_{1}}} \\
           & \left|\mathbf{a}_{2}^{\rm{H}}\mathbf{w}\right|\geq\sqrt b_{2}
\end{aligned}
\end{gather}

Substituting the expression of ${\bf{a}}_i$, we can rewrite the above problem as
\begin{gather} \label{eq_problem_sub2_1c}
\begin{aligned}
\mathop{\mathrm{Minimize}}\limits_{{\bf{w}}}~~~~&\mathop{\mathrm{Max}}\limits_{i}{\{[\mathbf{w}]_{i}[\mathbf{w}]_{i}^{*}\}}\\
\mathrm{Subject~to} ~~~&~{\mathrm{Re}(\sum\limits_{i=1}^{N}[\mathbf{w}]_{i}e^{j\theta_{i}}) \geq\sqrt {b_{1}}} \\
& \left|\sum\limits_{i=1}^{N}[\mathbf{w}]_{i}e^{j\eta_{i}}\right|\geq\sqrt {b_{2}}
\end{aligned}
\end{gather}
where $\theta_{i}=(i-1)\pi\Omega_{1}$, $\eta_{i}=(i-1)\pi\Omega_{2}$.

Problem \eqref{eq_problem_sub2_1c} is still not convex because of the absolute value operation in the second constraint. Thus, we can split it into a serial of convex optimization problems, i.e., we assume different phases for $\sum_{i=1}^{N}[\mathbf{w}]_{i}e^{j\eta_{i}}$ and obtain $M$ convex problems
\begin{gather}\label{eq_problem_sub2_1d}
\begin{aligned}
\mathop{\mathrm{Minimize}}\limits_{{\bf{w}}}~~~~&\mathop{\mathrm{Max}}\limits_{i}{[\mathbf{w}]_{i}[\mathbf{w}]_{i}^{*}}\\
\mathrm{Subject~to} ~~~&~{\mathrm{Re}(\sum\limits_{i=1}^{N}[\mathbf{w}]_{i}e^{j\theta_{i}}) \geq\sqrt {b_{1}}} \\
&~{\mathrm{Re}(\left(\sum\limits_{i=1}^{N}[\mathbf{w}]_{i}e^{j\eta_{i}}\right)e^{j\frac{m}{M}2\pi }) \geq\sqrt {b_{2}}}
\end{aligned}
\end{gather}
where $M$ is the number of total candidate phases $(m=1,2,\cdots,M)$. Each of these $M$ problems can be efficiently solved by using standard convex optimization tools. We select the solution with the minimal objective among the $M$ optimal solutions as the optimal solution $\mathbf{w}_0^\star$, and then we normalize it such that the vector has unit power: $\mathbf{w}_1^\star=\mathbf{w}_0^\star/\|\mathbf{w}_0^\star\|$.

As we have mentioned, the original purpose of the beamforming problem is to obtain ${\bf{w}}$ satisfying $|{\bf{h}}_i^{\rm{H}}{\bf{w}}|^2=c_i^\star$ ($i=1,2$) under the CM constraint. We formulate the beamforming problem as \eqref{eq_problem_sub2} to pursue a relatively easy solution. However, there is no guarantee that the obtained ${\bf{w}}_1^\star$ by solving \eqref{eq_problem_sub2} or \eqref{eq_problem_sub2_1} satisfies the CM constraint. Hence, after solving Problem \eqref{eq_problem_sub2_1}, we still need to do CM normalization, i.e., to normalize ${\bf{w}}_1^\star$ to $\mathbf{w}^\star$, so as to satisfy the CM constraint with the phases of its elements unchanged. The CM normalization is given by
\begin{equation}\label{eq_CMNorm}
 [\mathbf{w}^\star]_{k}=\frac{[\mathbf{w}_1^\star]_{k}}{\sqrt{N}|[\mathbf{w}_1^\star]_k|},~k=1,2,...,N
\end{equation}

Although the CM normalization is natural, if the moduli of the elements of ${\bf{w}}_1^\star$ are different with each other, the CM normalization would result in significant influence, and the finally achieved beam gains, i.e., $c_i$, would be far away from $c_i^\star$. In such a case, the sum rate performance would be not satisfactory. Hence, we need to evaluate the impact of the CM normalization on performance. To this end, we give the following theorem.
\begin{theorem} If $\mathbf{w}_{\rm{opt}}$ is the optimal solution of Problem \eqref{eq_problem_sub2_1d}, then, at least $(N-1)$ elements of $\mathbf{w}_{\rm{opt}}$ have the same modulus, and the remaining one element has a smaller modulus than the $(N-1)$ elements.
\end{theorem}

\begin{proof}
See Appendix A.
\end{proof}

According to Theorem 1, since $\mathbf{w}_0^\star$ is the optimal solution of Problem \eqref{eq_problem_sub2_1d}, and $\mathbf{w}_1^\star$ is the 2-norm normalization of $\mathbf{w}_0^\star$, at least $(N-1)$ elements of $\mathbf{w}_1^\star$ have the same modulus, and the remaining one element has a modulus no more than the $(N-1)$ elements, i.e., $0\leq|[\mathbf{w}_1^\star]_1|\leq\frac{1}{\sqrt{N}}$ and
\begin{equation}
\frac{1}{\sqrt{N}}\leq|[\mathbf{w}_1^\star]_2|=|[\mathbf{w}_1^\star]_3|=...=|[\mathbf{w}_1^\star]_N|\leq\frac{1}{\sqrt{N-1}}
\end{equation}

\begin{theorem} Given ${\bf{b}}$ an arbitrary CM vector, i.e., $|{\bf{b}}|={\bf{1}}$, $\big||{\bf{b}}^{\rm{H}}{\bf{w}}^\star|-|{\bf{b}}^{\rm{H}}{\bf{w}}_1^\star|\big|<\frac{2}{\sqrt{N}}$.
\end{theorem}

\begin{proof}
See Appendix B.
\end{proof}

According to Theorem 2, since ${\bf{a}}_i$ ($i=1,2$) is a CM vector, $\big||{\bf{a}}_i^{\rm{H}}{\bf{w}}^\star|-|{\bf{a}}_i^{\rm{H}}{\bf{w}}_1^\star|\big|<\frac{2}{\sqrt{N}}$, which means that the CM normalization has a limited influence on the desired beam gains, and the influence decreases when $N$ increases.

{
\subsection{Solution of the Original Problem}
Substituting the obtained normalized ${\bf{w^{\star}}}$ above into the original problem, i.e., \eqref{eq_problem}, we obtain
\begin{equation} \label{eq_problem_final}
\begin{aligned}
\mathop{\mathrm{Maximize}}\limits_{p_1}~~~~&\log_{2}(1+ \frac{c_{1}p_{1}}{\sigma^{2}})+\log_{2}(1+ \frac{c_{2}(P-p_{1})}{c_{2}p_{1}+\sigma^{2}})\\
\mathrm{Subject~to} ~~~ &\log_{2}(1+ \frac{c_{1}p_{1}}{\sigma^{2}}) \geq r_{1}\\
&\log_{2}(1+ \frac{c_{2}(P-p_{1})}{c_{2}p_{1}+\sigma^{2}}) \geq r_{2} \\
\end{aligned}
\end{equation}
where the beam gain $c_i=|{\bf{h}}_i^{\rm{H}}{\bf{w^{\star}}}|^2$ ($i=1,2$) are fixed values. Problem \eqref{eq_problem_final} is a single-variable optimization problem. The feasible region and the monotonicity of the objective function are distinct. We can easily obtain the final power allocation $\{p_1,~p_2\}$ by solving the problem above.

It is noteworthy that although the whole solution uses the simplified effective channel model defined in \eqref{eq_effChannel}, it can be used for any exact channel model, because an effective channel model can always be defined based on an exact channel model. Besides, since $|\frac{1}{N}{\bf{a}}(N,\omega_1)^{\rm{H}}{\bf{a}}(N,\omega_2)|$ is small when $|\omega_1-\omega_2|\geq 2/N$ \cite{xiao2016codebook}, it easy to obtain $|\bar{{\bf{h}}}_i^{\rm{H}}{\bf{w}}|\approx |{\bf{h}}_i^{\rm{H}}{\bf{w}}|$, because ${\bf{w}}$ is designed to steer towards the AoD of ${\bf{h}}_i$. In other words, due to the spatial sparsity of the channel model, the beamforming gain towards the AoD of the strongest MPC is affected little by the other MPCs. Thus, when a solution $\{p_1,~p_2,~{\bf{w}}\}$ is obtained using the proposed approach with an effective channel model, the solution is also feasible in general for the original problem in \eqref{eq_problem} when the exact channel model is used. This conclusion is further verified via simulations in Section IV.

}

\subsection{Generalization}
The proposed solution solves the joint power allocation and beamforming problem for 2-user mmWave-NOMA with a ULA. It is natural to consider whether it can be generalized for other types of antenna arrays or more users.

The extension of the solution to other types of antenna arrays, like uniform planar array and uniform circular array, is possible. It is crucial that Lemma 1 may have different forms for different types of arrays. Except Lemma 1, the whole solution does not involve the feature of a specific antenna array (e.g., the antenna spacing and the array shape); thus it can be used for different types of arrays.

{On the other hand, the idea of decomposing the original problem into two sub-problems still work in more-user cases. However, the solution of power and beam gain allocation sub-problem may not be directly used, because when the number of users is greater, the number of variables (including powers and beamforming gains) will be greater. In particular, there will be $(2K-2)$ independent variables in the power and beam gain allocation sub-problem for a $K$-user mmWave-NOMA system. Exhaustive search of the optimal solution may be used, but the complexity is $O((\frac{1}{\epsilon})^{2K-2})$, where $\epsilon$ is the searching precision. When $K$ is large, the complexity will be prohibitively high. On the other hand, the method to solve the 2-user beamforming sub-problem is instructive to find a similar solution for more-user cases. In the 2-user case, we need to search over $M$ possible phases for one user (see \eqref{eq_problem_sub2_1d}). Analogously, in a $K$-user case, we need to search over $M^{K-1}$ possible phases for $(K-1)$ users. In brief, if the number of users is not large, the idea of problem decomposition is still applicable to solve the original problem. However, when $K$ is large, the proposed solution may become not appropriate due to high complexity.

Fortunately, based on the proposed solution, there are other ways to support more users. For instance, one method is to combine with the OMA strategies to manyfold increase the number of users, or to use a hybrid beamforming structure with multiple RF chains, such that the number of users can be increased by $N_{\rm{RF}}$ times, where $N_{\rm{RF}}$ is the number of RF chains. Another method is to still use an analog beamforming structure and shape a few beams. The difference is that each beam steer towards a group of users rather than only one user in this paper. In such a case, we need to consider beam gain allocation between different user groups and power allocation within each user group. This topic will be studied in detail in our future work.}

\section{Performance Evaluations}
In this section, we evaluate the performance of the proposed joint power allocation and beamforming method. {As we consider a scenario that the NOMA users are located in different directions with a phased array, which is different from \cite{Ding2017random}, where the NOMA users are located in the same direction. On the other hand, in \cite{Daill2017}, a lens array was adopted to realize multi-beam forming to serve multiple NOMA users with arbitrary locations, where the power allocation problem is studied under fixed beam pattern. To the best of our knowledge, similar work is not found which considers joint power allocation and CM beamforming in mmWave-NOMA with a phased array. Hence, in this section we mainly make performance comparison between the proposed method and TDMA\footnote{Note that spatial-division multiple access (SDMA) is infeasible here, since an analog beamforming structure was used in the BS.}. In addition, we also compare the performance with the upper bound, which can reflect clearly how good the achieved performance is.} As aforementioned, the joint problem has been decomposed into two sub-problems, namely the power and beam gain allocation sub-problem and the beamforming sub-problem. For the power and beam gain allocation sub-problem, we find the optimal solution; while for the beamforming sub-problem, we find a sub-optimal solution, where the CM normalization operation may result in performance degradation. Hence, we start from the performance evaluation of the beamforming phase.


To compare the ideal beam pattern with the designed beam pattern obtained by solving Problem \eqref{eq_problem_sub2}, we assume $|\lambda_1|=0.8$, $|\lambda_1|=0.5$, $\Omega_1=-0.25$, $\Omega_2=0.4$. The desired beam gains are $c_1^\star=N/2$ and $c_2^\star=(N-c_1^\star/|\lambda_1|^2)|\lambda_2|^2$, where $N$ is the number of antennas at the BS. $M$ in \eqref{eq_problem_sub2_1d} is set to 20 in this simulation as well as the following simulations, which is large enough to obtain the best solution. As the CM normalization affects the shape of beam pattern, we compare the desired ideal beam pattern with the designed beam patterns before and after the CM normalization, i.e., the beam patterns computed with ${\bf{w}}_1^\star$ and ${\bf{w}}^\star$ in \eqref{eq_CMNorm}, respectively. Fig. \ref{fig:beam_pattern} shows the comparison results with $N=16,~32,~64$, and from this figure we can find that the beam gains are significant along the desired user directions, and both the beam patterns before and after the CM normalization are close to the ideal beam pattern along the user directions, which not only demonstrates that the CM normalization has little influence on the beam gains along the user directions, but also shows that the solution of the beamforming sub-problem is reasonable.

\begin{figure*}[t]
\begin{center}
  \includegraphics[width=19 cm]{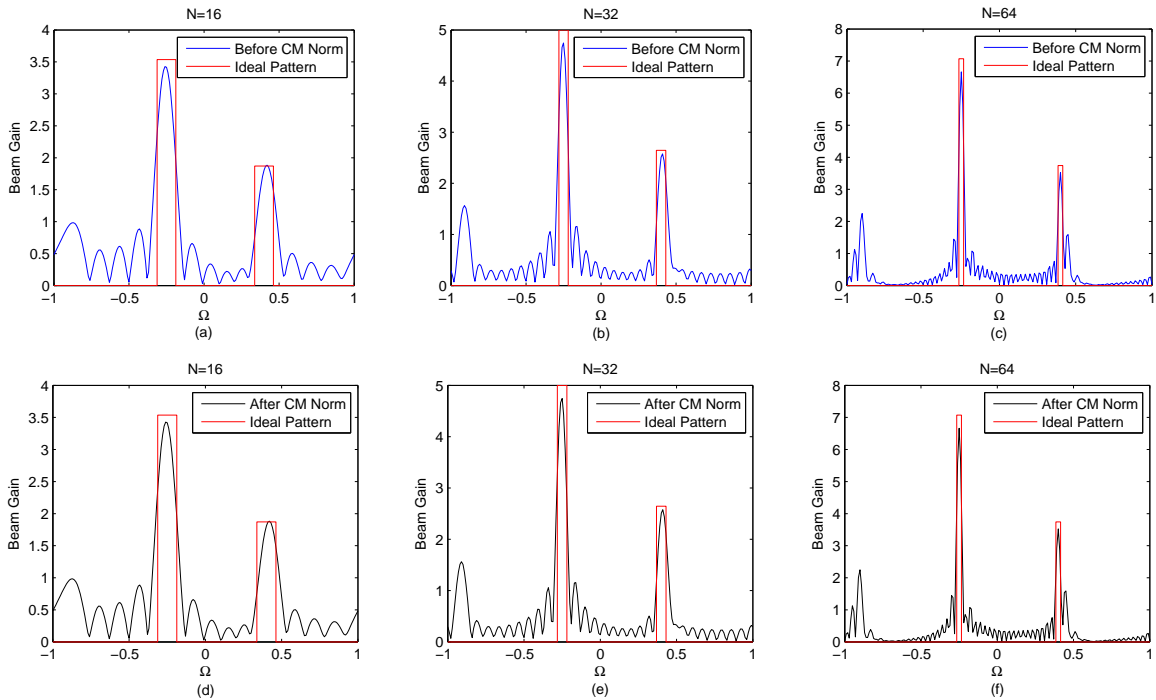}
  \caption{Comparison between the ideal beam pattern and the designed beam patterns before and after the CM normalization.}
  \label{fig:beam_pattern}
\end{center}
\end{figure*}

\begin{figure}[t]
\begin{center}
  \includegraphics[width=\figwidth cm]{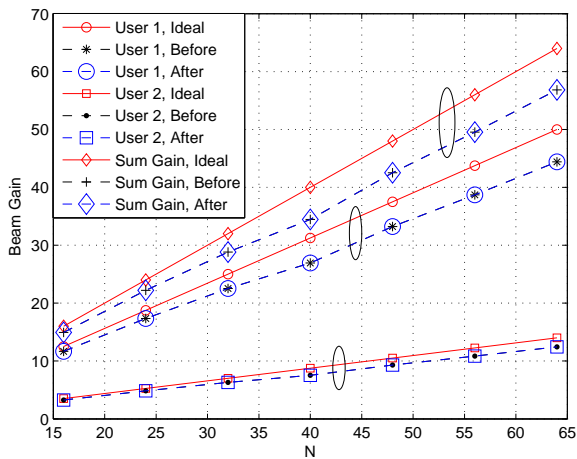}
  \caption{Comparison of user beam gains between the ideal beam gains and the designed beam gains before/after the CM normalization, where the sum gain refers to the summation of the beam gains of User 1 and User 2.}
  \label{fig:Gain_Cmp}
\end{center}
\end{figure}

\begin{figure}[t]
\begin{center}
  \includegraphics[width=\figwidth cm]{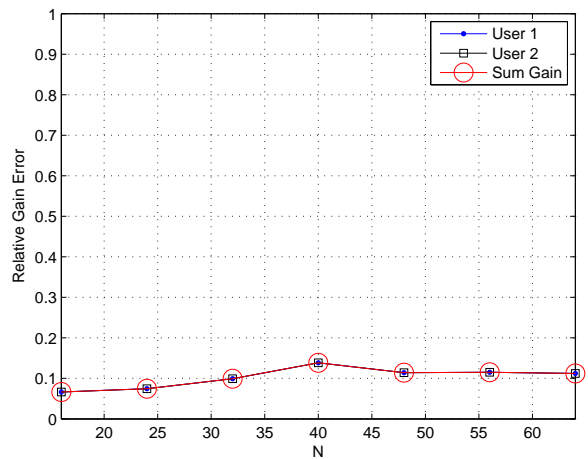}
  \caption{Relative gain errors versus the ideal beam gains of User 1, User 2 and the sum gain.}
  \label{fig:relative_gain_error}
\end{center}
\end{figure}

In addition to the beam pattern comparison, we also compare the user beam gains with varying number of antennas in Fig. \ref{fig:Gain_Cmp}, where the parameter settings are the same as those in Fig. \ref{fig:beam_pattern}. From Fig. \ref{fig:Gain_Cmp}, we can observe again that the user gains before and after the CM normalization are almost the same as each other, which demonstrates again that the CM normalization has little impact on the beamforming performance. Moreover, there is a small gap between the designed user gains and the ideal beam gains for both users (as well as the sum beam gain). This is because the designed beam pattern has side lobes which reduces the gains along the user directions. In comparison, an ideal beam pattern does not have side lobes. Fortunately, the gap increases slowly as $N$ increases when $N\leq 40$, and almost does not increase when $N>40$, which shows that the proposed beamforming method behaves robust against the number of antennas.

Fig. \ref{fig:relative_gain_error} shows the relative gain errors of User 1, User 2 and the sum gain versus the ideal/desired beam gains, where the parameter settings are the same as those in Fig. \ref{fig:Gain_Cmp}. Interestingly, from Fig. \ref{fig:relative_gain_error} we find that the relative beam gains of User 1 and User2, as well as the sum beam gain, are almost the same as each other. Moreover, the relative gain errors are small, roughly around 0.1, and they increase slowly as $N$ increases when $N\leq 40$, and almost does not increase when $N>40$. This result not only demonstrates again that the proposed beamforming method behaves robust against the number of antennas, but also shows the rational of Lemma 1, i.e., the sum beam gain can be roughly seen a constant versus $N$.

The above evaluations show that the solution of the beamforming sub-problem is reasonably close to the ideal one. Next, we evaluate the overall performance. Fig. \ref{fig:AR_Rate} shows the comparison between the performance bound and the designed achievable rates with varying rate constraint. The performance bound refers to the achievable rate obtained by solving only the power and beam gain allocation sub-problem \eqref{eq_problem_sub1}, i.e., with parameters $\{c_1^\star,~c_2^\star,~p_1^\star,~p_2^\star\}$, where the beamforming is assumed ideal. The designed performance refers to the achievable rate obtained by solving the original problem \eqref{eq_problem}. Relevant parameter settings are $\sigma^2=1$ mW, $P=100$ mW, $N=32$, $|\lambda_1|=0.8$, $|\lambda_1|=0.5$, $\Omega_1=-0.25$, $\Omega_2=0.4$. From Fig. \ref{fig:AR_Rate} we can find that the designed achievable rates are close to the achievable-rate bound for both User 1 and User 2, as well as the sum rate, which demonstrates that the proposed solution to the original problem is rational and effective, i.e., it can achieve close-to-bound performance. On the other hand, we can find that most power or beam gain is allocated to User 1, which has the better channel condition, so as to optimize the sum rate. Only necessary power or beam gain is allocated to User 2 to satisfy the rate constraint. That is why User 2 always achieves an achievable rate equal to the rate constraint.

\begin{figure}[t]
\begin{center}
  \includegraphics[width=\figwidth cm]{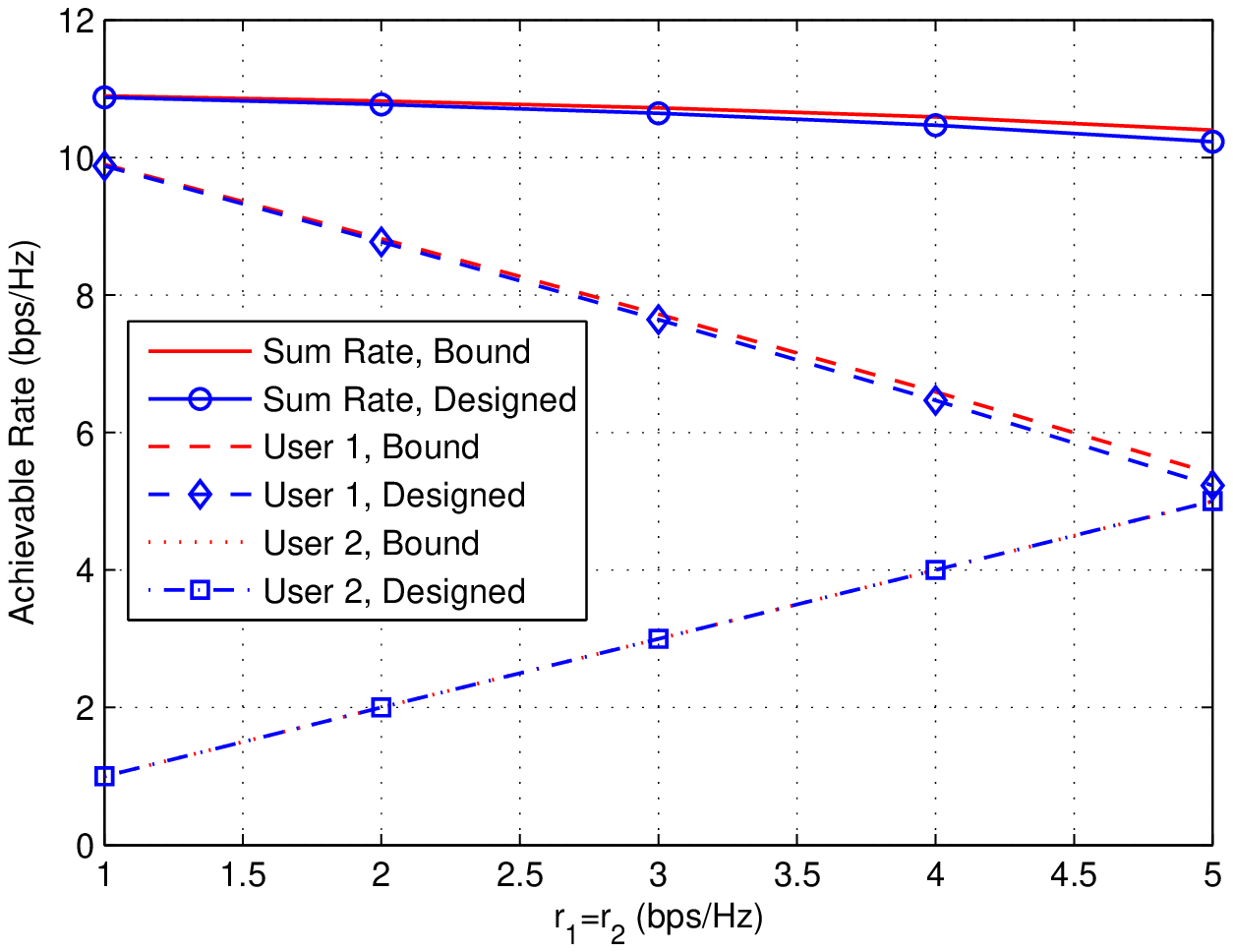}
  \caption{Comparison between the performance bound and the designed achievable rates with varying rate constraint.}
  \label{fig:AR_Rate}
\end{center}
\end{figure}

\begin{figure}[t]
\begin{center}
  \includegraphics[width=\figwidth cm]{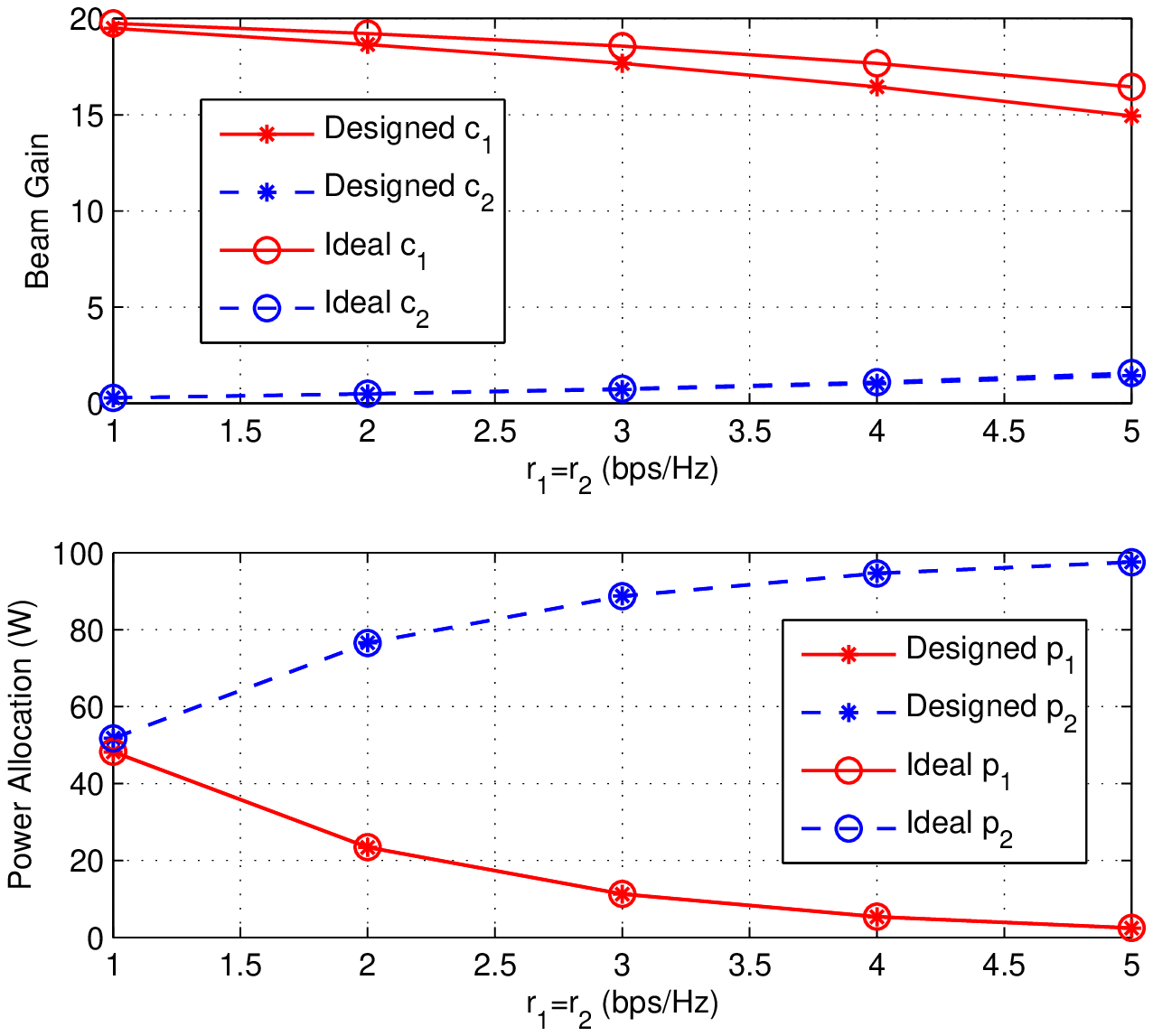}
  \caption{Comparison between the ideal values of the parameters and the designed values with varying rate constraint.}
  \label{fig:Parameters_rate}
\end{center}
\end{figure}

Fig. \ref{fig:Parameters_rate} shows the comparison between the ideal values of the parameters and the designed values with varying rate constraint. The ideal values $\{c_1,~c_2,~p_1,~p_2\}$ refer to $\{c_1^\star,~c_2^\star,~p_1^\star,~p_2^\star\}$, i.e., the optimal solution of the power and beam gain allocation sub-problem \eqref{eq_problem_sub1}, while the designed $\{c_1,~c_2\}$ refer to $\{|{\bf{h}}_1^{\rm{H}}{\bf{w}}^\star|^2,~|{\bf{h}}_2^{\rm{H}}{\bf{w}}^\star|^2\}$, i.e., the beam gains achieved by the final solution to the original problem \eqref{eq_problem}. The parameter settings are the same as those in Fig. \ref{fig:AR_Rate}. From Fig. \ref{fig:Parameters_rate} we can find that the designed beam gains are close to the ideal gains. The gap between the designed gain and the ideal gain for User 1 is due to the fact that there are side lobes for the designed beam pattern, but for the ideal beam pattern there is no side lobe. More importantly, for User 1, which has a better channel condition, the beam gain is much higher than User 2, while the allocated power is smaller than User 2. This result is the same with the conventional 2-user NOMA system, where necessary power should be allocated to the user with a worse channel condition to satisfy the rate constraint, and the rest power is allocated to the better one to maximize the sum rate. Also, we can observe that as the rate constraint increases, the beam gain and power of User 1 decrease, while those of User 2 increase, but the varying speed of beam gain is much slower than that of power for both users.

Fig. \ref{fig:AR_Power} shows the comparison between the performance bound and the designed achievable rates with varying total power to noise ratio. Relevant parameter settings are $N=32$, $|\lambda_1|=0.8$, $|\lambda_2|=0.5$, $\Omega_1=-0.25$, $\Omega_2=0.4$, $r_1=r_2=3$ bps/Hz. From this figure we can observe the similar results as those from Fig. \ref{fig:AR_Rate}, i.e., the designed achievable rates are close to the performance bounds for both User 1 and User 2, as well as the sum rate, and most power or beam gain is allocated to User 1 to optimize the sum rate, while only necessary power or beam gain is allocated to User 2 to satisfy the rate constraint.

\begin{figure}[t]
\begin{center}
  \includegraphics[width=\figwidth cm]{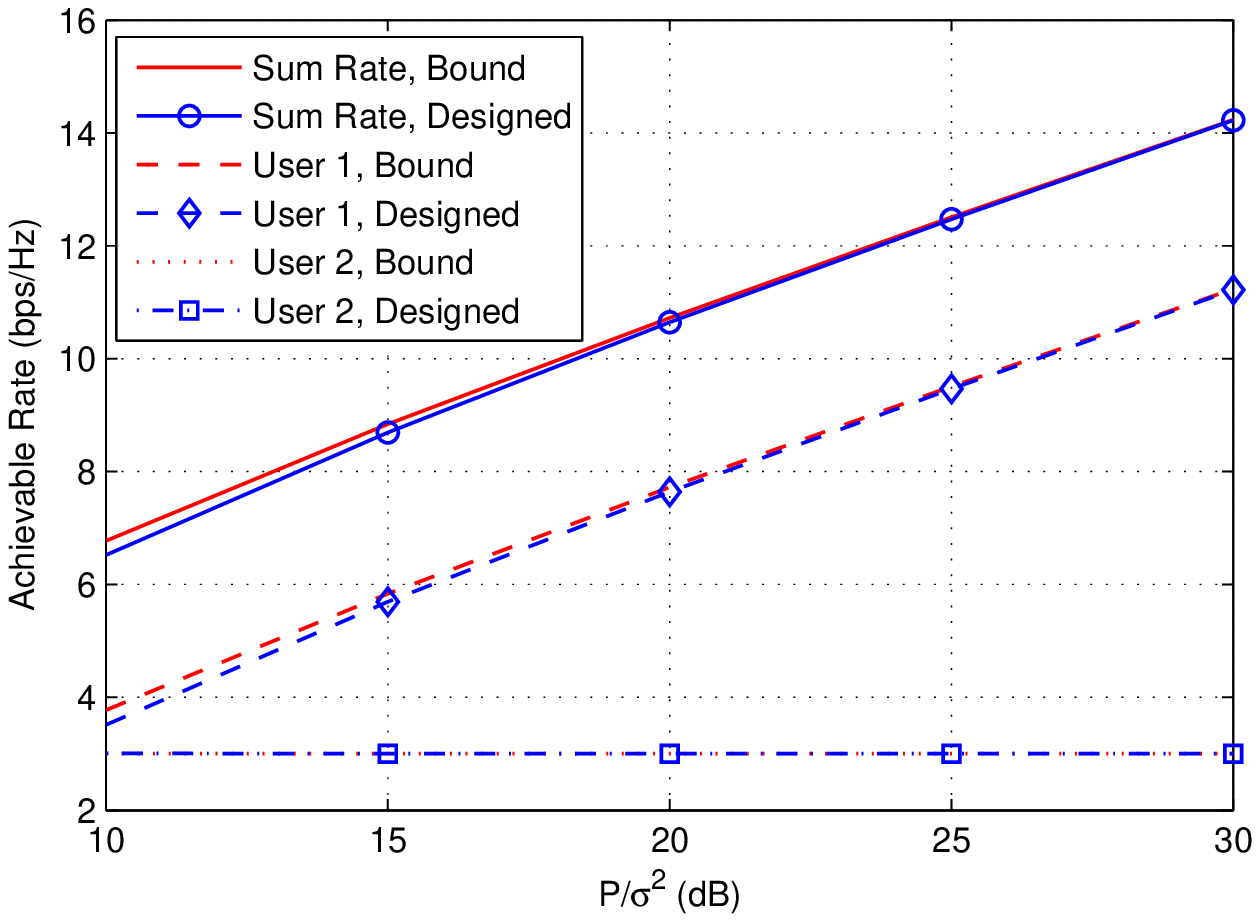}
  \caption{Comparison between the performance bound and the designed achievable rates with varying total power to noise ratio.}
  \label{fig:AR_Power}
\end{center}
\end{figure}

Fig. \ref{fig:Parameters_Power} shows the comparison between the ideal values of the parameters and the designed values with varying total power to noise ratio. The parameter settings are the same as those in Fig. \ref{fig:AR_Power}, and $\sigma^2=1$ mW here. Again, we can find that the designed beam gains are close to the ideal gains. For User 1, the beam gain is much higher than User 2, while the allocated power is smaller than User 2. Also, we can observe that as $P/\sigma^2$ increases, the beam gain and power of User 1 increase, while the beam gain of User 2 decreases on the contrary.

From Figs. \ref{fig:Parameters_rate} and \ref{fig:Parameters_Power}, we can find that for User 2, the beam gain is small in general, and varies slowly as the rate constraint and the total power to noise ratio increases. This is because, as shown in \eqref{eq_problem_sub1}, when increasing the beam gain the interference from User 1 also increases, but when increasing the power the interference does not increases. Hence, for User 2 the beamforming gain is small in general. Most of the beam gain is allocated to User 1, because for User 1 the interference from User 2 can be decoded and removed by using SIC.

\begin{figure}[t]
\begin{center}
  \includegraphics[width=\figwidth cm]{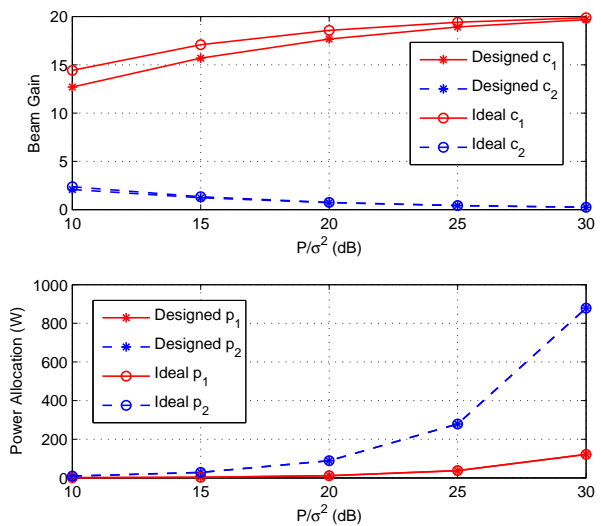}
  \caption{Comparison between the ideal values of the parameters and the designed values with varying total power to noise ratio.}
  \label{fig:Parameters_Power}
\end{center}
\end{figure}

We have demonstrated that the proposed joint power allocation and beamforming approach works well for the mmWave-NOMA system, and can achieve close-to-bound sum-rate performance. Next, we want to compare the performance of mmWave-NOMA with that of TDMA. In addition, we have used the effective channel in \eqref{eq_effChannel} in the analysis and derivation. In practice, the original channel model in \eqref{eq_oriChannel} should be adopted. Hence, we need to compare the practical performance of the proposed solution under the original channel and the theoretical performance under the effective channel.

Figs. \ref{fig:AR_Cmp_Rate} and \ref{fig:AR_Cmp_Power} show the comparison results of sum rate between theoretical mmWave-NOMA, practical mmWave-NOMA and TDMA with varying rate constraint and varying total power to noise ratio, respectively, where $N=32$ and $L_1=L_2=L=4$. Note that the theoretical mmWave-NOMA is the performance under the effective channel model \eqref{eq_effChannel} and the practical mmWave-NOMA is the performance under the original channel model \eqref{eq_oriChannel}. User 1 has a better channel condition than User 2, i.e., the average power ratio of them is $(1/0.3)^2$. For Fig. \ref{fig:AR_Cmp_Rate}, $\sigma^2=1$ mW and $P=100$ mW, while for Fig. \ref{fig:AR_Cmp_Power} $r_1=r_2=3$ bps/Hz. Both LOS and NLOS channel models are considered. For LOS channel, the first path is the LOS path, which has a constant power\footnote{The LOS component may also be modeled as Rayleigh fading, but the performance with a LOS channel will be similar to that with an NLOS channel.}, i.e., $|\lambda_1|=1$ (0 dB), while the coefficients of the other 3 NLOS paths, i.e., $\{\lambda_i\}_{i=2,3,4}$, obey the complex Gaussian distribution with zero mean, and each of them has an average power of -10/-15 dB. For the NLOS channel, the 4 paths are all NLOS paths with zero-mean complex Gaussian distributed coefficients, and each of them has an average power of $1/\sqrt{L}$. Each point in Figs. \ref{fig:AR_Cmp_Rate} and \ref{fig:AR_Cmp_Power} is the average performance based on $10^3$ channel realizations. With each channel realization, the optimal parameters are obtained by the proposed solution, and the theoretical/practical performances are obtained by computing the sum rates with the effective/original channel. The performance of TDMA is obtained based on the assumption that the beam gains of User 1 and User 2 are equal, i.e., $N/2$. From these two figures we can observe that the theoretical performance is very close to the practical performance, which demonstrates the rational of the proposed method. Moreover, the performance of mmWave-NOMA is significantly better than that of TDMA under both the LOS and NLOS channels.

\begin{figure}[t]
\begin{center}
  \includegraphics[width=\figwidth cm]{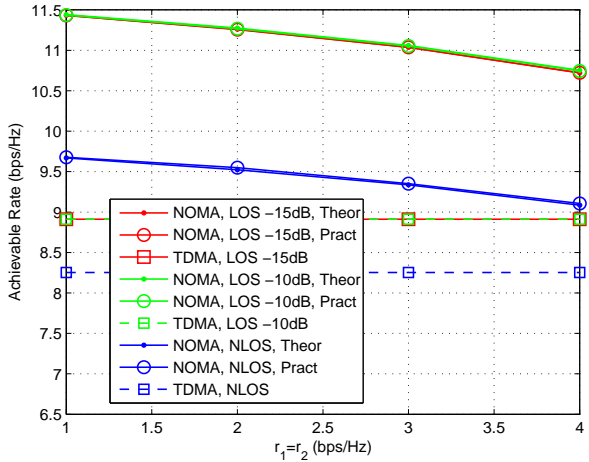}
  \caption{Comparison of sum rate between theoretical mmWave-NOMA, practical mmWave-NOMA and TDMA with varying rate constraint.}
  \label{fig:AR_Cmp_Rate}
\end{center}
\end{figure}

\begin{figure}[t]
\begin{center}
  \includegraphics[width=\figwidth cm]{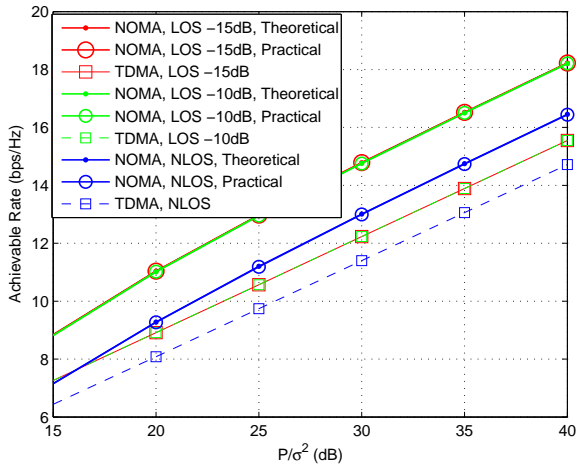}
  \caption{Comparison of sum rate between theoretical mmWave-NOMA, practical mmWave-NOMA and TDMA with varying total power to noise ratio.}
  \label{fig:AR_Cmp_Power}
\end{center}
\end{figure}

\section{Conclusion}
In this paper we have investigated the problem of how to maximize the sum rate of a 2-user mmWave-NOMA system, where we need to find the beamforming vector to steer towards the two users simultaneously subject to an analog beamforming structure, and meanwhile allocate appropriate power to them. We have proposed a suboptimal solution to this problem, i.e., to decompose the original problem into two sub-problems: one is a power and beam gain allocation problem, and the other is a beamforming problem under the CM constraint. The original problem can then be solved by solving the two sub-problems. Extensive performance evaluations verify the rational of the proposed solution, and show that the solution can achieve close-to-bound sum-rate performance, which is distinctively better than TDMA.


%

\appendices
\section{Proof of Theorem 1}

Let $\mathbf{w}_{\rm{opt}}$ represent the optimal solution of Problem \eqref{eq_problem_sub2_1d}, and assume
\begin{equation}
\left\{\begin{aligned}
  &\sum\limits_{i=1}^{N}[\mathbf{w}_{\rm{opt}}]_{i}e^{j\theta_{i}}=d_{1}
  \\&\sum\limits_{i=1}^{N}[\mathbf{w}_{\rm{opt}}]_{i}e^{j\eta_{i}}=d_{2}e^{-j\frac{m}{M}2\pi }
\end{aligned}\right.
\end{equation}
where $d_1$ and $d_2$ are positive real values.

\begin{lemma}
Given $d_{1}, d_{2}$, Problem \eqref{eq_problem_sub2_1d} is equivalent to
\begin{gather} \label{eq_problem_sub2_1b}
\begin{aligned}
\mathop{\mathrm{Minimize}}\limits_{{\bf{w}}}~~~~&\mathop{\mathrm{Max}}\limits_{i}{\{[\mathbf{w}]_{i}[\mathbf{w}]_{i}^{*}\}}\\
\mathrm{Subject~to} ~~~&~\sum\limits_{i=1}^{N}[\mathbf{w}]_{i}e^{j\theta_{i}}=d_{1}\\
                    &~\sum\limits_{i=1}^{N}[\mathbf{w}]_{i}e^{j\eta_{i}}=d_{2}e^{-j\frac{m}{M}2\pi }
\end{aligned}
\end{gather}
\end{lemma}

\begin{proof}
According to the definitions of $d_{1}$ and $d_{2}$, the optimal solution of Problem \eqref{eq_problem_sub2_1d}, i.e., $\mathbf{w}_{\rm{opt}}$, is a feasible solution of Problem \eqref{eq_problem_sub2_1b}.

On the other hand, since $d_1\geq \sqrt{b_1}$ and $d_2\geq \sqrt{b_2}$, the optimal solution of Problem \eqref{eq_problem_sub2_1b} must be a feasible solution of Problem \eqref{eq_problem_sub2_1d}.

In summary, Problem \eqref{eq_problem_sub2_1b} is equivalent to Problem \eqref{eq_problem_sub2_1d}.
\end{proof}


\begin{lemma}
 There are at least $(N-1)$ elements of ${\bf{w}}_{\rm{opt}}$ which have the same modulus, where ${\bf{w}}_{\rm{opt}}$ is the optimal solution of Problem \eqref{eq_problem_sub2_1b}, and the remaining one element has a smaller modulus than the $(N-1)$ elements.
\end{lemma}

\begin{proof}
We first rank the absolute weights of $\mathbf{w}_{\rm{opt}}$ as
\begin{equation} \label{eq_inequalities}
  |[\mathbf{w}_{\rm{opt}}]_{\pi_{1}}|\leq|[\mathbf{w}_{\rm{opt}}]_{\pi_{2}}|\leq|[\mathbf{w}_{\rm{opt}}]_{\pi_{3}}|
  \leq\cdots\leq|[\mathbf{w}_{\rm{opt}}]_{\pi_{N}}|
\end{equation}
Then Lemma 3 is equivalent to that the inequalities after $|[\mathbf{w}_{\rm{opt}}]_{\pi_{2}}|$ are all equalities. We prove it by using the contradiction method, i.e., we assume that there is at least one strictly less-than sign after $|[\mathbf{w}_{\rm{opt}}]_{\pi_{2}}|$, and then we prove that $\mathbf{w}_{\rm{opt}}$ is not optimal.

Since the two constraints of Problem \eqref{eq_problem_sub2_1b} are two linear equations in the space $\mathbb{C}^N$, $[\mathbf{w}_{\rm{opt}}]_{\pi_{1}}$ and $[\mathbf{w}_{\rm{opt}}]_{\pi_{2}}$ can be expressed as linear combinations of $\{[\mathbf{w}_{\rm{opt}}]_{\pi_{k}}\}_{k=3,4,...,N}$, which means that the feasible region of Problem \eqref{eq_problem_sub2_1b} is $\mathbb{C}^{N-2}$. Without loss of generality, let ${\bf{w}}_0\triangleq [[\mathbf{w}_{\rm{opt}}]_{\pi_{3}},[\mathbf{w}_{\rm{opt}}]_{\pi_{4}},...,[\mathbf{w}_{\rm{opt}}]_{\pi_{N}}]^{\rm{T}}$, and let $[\mathbf{w}_{\rm{opt}}]_{\pi_{1}}={\bf{f}}_1^{\rm{H}}{\bf{w}}_0+\beta_1$, $[\mathbf{w}_{\rm{opt}}]_{\pi_{2}}={\bf{f}}_2^{\rm{H}}{\bf{w}}_0+\beta_2$, where ${\bf{f}}_1$ and ${\bf{f}}_2$ are the combination coefficient vectors, $\beta_1$ can $\beta_2$ are constants. In the following, we will construct $\bar{\bf{w}}$, a better solution than $\mathbf{w}_{\rm{opt}}$.

We consider a point in the feasible region $\mathbb{C}^{N-2}$ close to ${\bf{w}}_0$, i.e., $\bar{\bf{w}}_0=\frac{1}{1+\delta}{\bf{w}}_0$, where $\delta$ is a small positive variable that is very close to zero. Let $w_1={\bf{f}}_1^{\rm{H}}\bar{\bf{w}}_0+\beta_1$ and $w_2={\bf{f}}_2^{\rm{H}}\bar{\bf{w}}_0+\beta_2$. Then $\bar{\bf{w}}=[w_1,w_2,\bar{\bf{w}}_0^{\rm{T}}]^{\rm{T}}$ is a feasible point of Problem \eqref{eq_problem_sub2_1b}. We have

\begin{equation}
\left\{\begin{aligned}
  &|w_1|=|{\bf{f}}_1^{\rm{H}}\bar{\bf{w}}_0+\beta_1|\\
&|w_2|=|{\bf{f}}_2^{\rm{H}}\bar{\bf{w}}_0+\beta_2|
\end{aligned}\right.
\end{equation}

When there is at least one strictly less-than sign after $|[\mathbf{w}_{\rm{opt}}]_{\pi_{2}}|$ in \eqref{eq_inequalities}, we have $|[\mathbf{w}_{\rm{opt}}]_{\pi_{1}}|\leq |[\mathbf{w}_{\rm{opt}}]_{\pi_{2}}|<|[\mathbf{w}_{\rm{opt}}]_{\pi_{N}}|$.

Suppose $|[\mathbf{w}_{\rm{opt}}]_{\pi_{N}}|-|[\mathbf{w}_{\rm{opt}}]_{\pi_{1}}|=\varepsilon_1$, where $\varepsilon_1>0$. In a sequel,
\begin{equation}
\begin{aligned}
&|[\bar{\bf{w}}]_{\pi_{N}}|-|w_1|\\
=&\frac{1}{1+\delta}|[\mathbf{w}_{\rm{opt}}]_{\pi_{N}}|-\frac{1}{1+\delta}|[\mathbf{w}_{\rm{opt}}]_{\pi_{1}}|\\
&~~+\frac{1}{1+\delta}|[\mathbf{w}_{\rm{opt}}]_{\pi_{1}}|-|w_1|\\
=&\frac{\varepsilon_1}{1+\delta}+\frac{1}{1+\delta}|{\bf{f}}_1^{\rm{H}}{\bf{w}}_0+\beta_1|-|{\bf{f}}_1^{\rm{H}}\bar{\bf{w}}_0+\beta_1|\\
=&\frac{\varepsilon_1}{1+\delta}+|\frac{1}{1+\delta}{\bf{f}}_1^{\rm{H}}{\bf{w}}_0+\frac{1}{1+\delta}\beta_1|-|\frac{1}{1+\delta}{\bf{f}}_1^{\rm{H}}{\bf{w}}_0+\beta_1|\\
\geq & \frac{\varepsilon_1}{1+\delta}-|(\frac{1}{1+\delta}{\bf{f}}_1^{\rm{H}}{\bf{w}}_0+\frac{1}{1+\delta}\beta_1)-(\frac{1}{1+\delta}{\bf{f}}_1^{\rm{H}}{\bf{w}}_0+\beta_1)|\\
=&\frac{\varepsilon_1-\delta |\beta_1|}{1+\delta}
\end{aligned}
\end{equation}
Hence, there exists a sufficiently small $\delta_1=\frac{\varepsilon_1}{1+|\beta_1|}$ such that $\frac{\varepsilon_1-\delta_1 |\beta_1|}{1+\delta}>0$, i.e. $|w_1|<|[\bar{\bf{w}}]_{\pi_{N}}|$. Similarly, Supposing $|[\mathbf{w}_{\rm{opt}}]_{\pi_{N}}|-|[\mathbf{w}_{\rm{opt}}]_{\pi_{2}}|=\varepsilon_2$, we can conclude that there exists a sufficiently small $\delta_2=\frac{\varepsilon_1}{1+|\beta_2|}$ such that $|w_2|<|[\bar{\bf{w}}]_{\pi_{N}}|$. Let $\delta=\min\{\delta_1,\delta_2\}$, so there is always $|w_1|<|[\bar{\bf{w}}]_{\pi_{N}}|$ and $|w_2|<|[\bar{\bf{w}}]_{\pi_{N}}|$. In other words, $[\bar{\bf{w}}]_{\pi_{N}}$ is the largest-modulus element of $\bar{\bf{w}}$. Therefore,
\begin{equation}
\mathop{\mathrm{Max}}\limits_{i}{\{[\bar{\bf{w}}]_{i}[\bar{\bf{w}}]_{i}^{*}\}}=|[\bar{\bf{w}}]_{\pi_{N}}|^2<|[\mathbf{w}_{\rm{opt}}]_{\pi_{N}}|^2
\end{equation}
which means that $\bar{\bf{w}}$ is a better solution of Problem \eqref{eq_problem_sub2_1b} than $\mathbf{w}_{\rm{opt}}$. This is contradictory against that $\mathbf{w}_{\rm{opt}}$ is the optimal solution of Problem \eqref{eq_problem_sub2_1b}; so the assumption that there is at least one strictly less-than sign after $|[\mathbf{w}_{\rm{opt}}]_{\pi_{2}}|$ in \eqref{eq_inequalities} does not hold. Hence, the inequalities after $|[\mathbf{w}_{\rm{opt}}]_{\pi_{2}}|$ in \eqref{eq_inequalities} are all equalities, i.e., Lemma 3 holds.
\end{proof}

Combining Lemma 3 and Lemma 2, we can easily conclude that Theorem 1 holds.

\section{Proof of Theorem 2}
\begin{equation}
\begin{aligned}
&\big||{{\bf{b}}^{\rm{H}}}{{\bf{w}}^ \star }| - |{{\bf{b}}^{\rm{H}}}{\bf{w}}_1^ \star |\big|\le\big|{{\bf{b}}^{\rm{H}}}{{\bf{w}}^ \star } - {{\bf{b}}^{\rm{H}}}{\bf{w}}_1^ \star \big|\\
 \le& \big|[{\bf{b}}]_1^*({[{{\bf{w}}^ \star }]_1} - {[{\bf{w}}_1^ \star ]_1})\big| + \big|[{\bf{b}}]_2^*({[{{\bf{w}}^ \star }]_2} - {[{\bf{w}}_1^ \star ]_2})\big| + \cdot\cdot\cdot+\\
&\big|[{\bf{b}}]_N^*({[{{\bf{w}}^ \star }]_N} - {[{\bf{w}}_1^ \star ]_N})\big|\\
 =& \big|{[{{\bf{w}}^ \star }]_1} - {[{\bf{w}}_1^ \star ]_1}\big| + \big|{[{{\bf{w}}^ \star }]_2} - {[{\bf{w}}_1^ \star ]_2}\big| +\cdot\cdot\cdot+\\
&\big|{[{{\bf{w}}^ \star }]_N} - {[{\bf{w}}_1^ \star ]_N}\big|
\end{aligned}
\end{equation}

According to the CM normalization in \eqref{eq_CMNorm}, the elements of ${\bf{w}}^ \star$ have exactly the same phases as those of ${\bf{w}}_1^ \star$; hence
\begin{equation}
\begin{aligned}
&\big|{[{{\bf{w}}^ \star }]_i} - {[{\bf{w}}_1^ \star ]_i}\big|=\big||{[{{\bf{w}}^ \star }]_i}| - |{[{\bf{w}}_1^ \star ]_i}|\big|\\
=&\big|\frac{1}{\sqrt{N}} - |{[{\bf{w}}_1^ \star ]_i}|\big|
\end{aligned}
\end{equation}
where $i=1,2,...,N$.

Since $0\leq|[{\bf{w}}_1^ \star ]_1|\leq\frac{1}{\sqrt{N}}$, we have $\big|{[{{\bf{w}}^ \star }]_1} - {[{\bf{w}}_1^ \star ]_1}\big|\leq \frac{1}{\sqrt{N}}$. Since $\frac{1}{\sqrt{N}}\leq|[\mathbf{w}_1^\star]_2|=|[\mathbf{w}_1^\star]_3|=...=|[\mathbf{w}_1^\star]_N|\leq\frac{1}{\sqrt{N-1}}$, we have
\begin{equation}
\begin{aligned}
&\big|{[{{\bf{w}}^ \star }]_i} - {[{\bf{w}}_1^ \star ]_i}\big|\leq \frac{1}{\sqrt{N-1}}-\frac{1}{\sqrt{N}}
\end{aligned}
\end{equation}
where $i=2,3,...,N$.

Hence, we obtain
\begin{equation}
\begin{aligned}
&\big||{{\bf{b}}^{\rm{H}}}{{\bf{w}}^ \star }| - |{{\bf{b}}^{\rm{H}}}{\bf{w}}_1^ \star |\big|\\
\leq&\frac{1}{\sqrt{N}}+(N-1)\left(\frac{1}{\sqrt{N-1}}-\frac{1}{\sqrt{N}}\right)\\
<&\frac{2}{\sqrt{N}}
\end{aligned}
\end{equation}


\end{document}